\documentclass[preprint, 12pt]{elsarticle}

\usepackage{amsmath}
\usepackage{amssymb}
\usepackage{amsthm}
\usepackage{bbm}
\usepackage{algorithm}
\usepackage{algorithmic}
\usepackage{verbatim}
\usepackage{tikz}
\usepackage{graphicx}
\usepackage{subcaption}
\usetikzlibrary{positioning}
\usetikzlibrary{arrows}
\usepackage{verbatim}
\usepackage{cite}
\usepackage{hyperref}

\journal{Theoretical Computer Science}

\newcommand{\A}{\mathcal{A}}

\renewcommand{\P}{\mathbf{P}}

\renewcommand{\b}[1]{( #1 )}
\newcommand{\B}[1]{\{ #1 \}}
\renewcommand{\sb}[1]{[ #1 ]}

\newcommand{\R}{\mathbb{R}}
\newcommand{\N}{\mathbb{N}}

\newcommand{\1}{\mathbbm{1}}

\newcommand{\EE}{\mathcal{E}}
\newcommand{\fl}[1]{\lfloor #1 \rfloor}
\newcommand{\X}{\mathbf{X}}
\newcommand{\PPi}{\mathbf{\Pi}}
\newcommand{\F}{\mathcal{F}}

\newcommand{\W}{\mathbf{W}}

\newcommand{\Wu}[1]{#1}
\newcommand{\WWu}[1]{{#1}}

\newcommand{\Wurevision}[1]{#1}
\newcommand{\SecondRevision}[1]{{\color{blue}#1}}

\newcommand{\rightnote}[1]{}
\newcommand{\leftnote}[1]{}

\newtheorem{theorem}{Theorem}
\newtheorem{corollary}{Corollary}

\newtheorem{lemma}{Lemma}
\newtheorem{claim}{Claim}

\begin{document}

\begin{frontmatter}

\title{Stochastic Runtime Analysis of a  Cross-Entropy Algorithm for Traveling Salesman Problems}





\author[zw]{Zijun Wu\corref{cor1}\fnref{fn1}}
\ead{zijunwu@bjut.edu.cn}

\author[rm]{Rolf H. M{\"o}hring\corref{cor2}\fnref{fn2}}
\ead{rolf.moehring@tu-berlin.de}

\author[lj]{Jianhui Lai\corref{cor3}}
\ead{laijianhui@bjut.edu.cn}

\cortext[cor1]{Principal corresponding author}
\cortext[cor2]{Corresponding author}
\cortext[cor3]{Corresponding author}

\address[zw]{Beijing Institute for Scientific and Engineering Computing (BISEC), Pingle Yuan 100, Beijing, P. R. China}
\address[rm]{Beijing Institute for Scientific and Engineering Computing (BISEC), Pingle Yuan 100, Beijing, P. R. China}
\address[lj]{College of Metropolitan Transportation, Beijing University of Technology, Pingle Yuan 100, Beijing, P. R. China}

\fntext[fn1]{The author is also affiliated with the School of Applied Mathematics and Physics at Beijing University of Technology}
\fntext[fn2]{The author is a professor emeritus of mathematics at Berlin University
	of Technology}

\begin{abstract}
This article 
analyzes the stochastic runtime of a Cross-Entropy Algorithm \Wurevision{mimicking an Max-Min Ant System with iteration-best reinforcement. It investigates the impact of magnitude of the sample size on the runtime to find optimal solutions for TSP instances.}

For simple \Wurevision{TSP} instances that have a $\{1,n\}$-valued distance function and a unique optimal solution,  we \Wurevision{show that sample size $N\in \omega(\ln n)$ results in a stochastically polynomial runtime, and $N \in O(\ln n)$ results in a stochastically exponential runtime, where
``stochastically" means with a probability
of $1-n^{-\omega(1)},$ and $n$ represents number of cities. In particular, for $N\in \omega(\ln n),$ we prove
a stochastic runtime} of \Wurevision{$O\b{N\cdot n^{6}}$} with the vertex-based random solution generation,
and a stochastic runtime of \Wurevision{$O\b{N\cdot n^{3}\ln n}$} with the edge-based random solution generation.
These runtimes \Wurevision{are very close to} the \Wurevision{best} known expected runtime for variants of Max-Min Ant System with best-so-far reinforcement \Wurevision{by choosing a small $N\in \omega(\ln n).$}
They are obtained for the stronger notion of stochastic runtime, \Wurevision{and analyze the runtime in most cases.}

We also inspect more complex instances with $n$ vertices positioned on an $m\times m$ grid.
When the $n$ vertices span a convex polygon, we obtain a stochastic runtime of \Wurevision{$O(n^{4}m^{5+\epsilon})$} with the vertex-based random solution generation,
and a stochastic runtime of \Wurevision{$O(n^{3}m^{5+\epsilon})$} for the edge-based random solution generation.
When there are \Wurevision{$k \in O\b{1}$} many vertices inside a convex polygon spanned by the other $n-k$ vertices, we obtain
a stochastic runtime of \Wu{$O(n^{4}m^{5+\epsilon}+n^{6k-1}m^{\epsilon})$} with the vertex-based random solution
generation, and a stochastic runtime of \Wu{$O(n^{3}m^{5+\epsilon}+n^{3k}m^{\epsilon})$} with the edge-based 
random solution generation. These runtimes
are better than the expected runtime for the so-called $(\mu\!+\!\lambda)$ EA reported in a recent article, and again obtained for the stronger notion of stochastic runtime.
\end{abstract}


\end{frontmatter}


%
%



\section{Introduction}
The Cross 
Entropy (CE) algorithm
is a general-purpose evolutionary algorithm (EA) that has been applied successfully to many 
$\mathcal{NP}$-hard combinatorial optimization problems,
see e.g. the book \citep{rubinstein2013cross} for an overview.
It was initially designed for rare event simulation  by Rubinstein \citep{Rubinstein1997}
in 1997, and thereafter 
formulated as
an optimization tool for both continuous and discrete optimization (see \citep{Rubinstein1999}).


CE has much in common with the famous ant colony optimization (ACO, see \citep{dorigobook}) and 
the estimation of distribution algorithms (EDAs, see \citep{hauschild2011introduction}). 
They all belong to the so-called {\em model-based search} paradigm (MBS),
see \citep{zlochin2004model}. Instead of only manipulating
 solutions, which is very typical in traditional heuristics 
like Genetic Algorithms \citep{whitley1994genetic}  and Local Search \citep{lourencco2003iterated} and others, MBS algorithms attempt to optimize the solution reproducing mechanism. 
In each iteration, they produce
new solutions by sampling from a probabilistic distribution 
on the search space. 
The distribution is often called a {\em model} in the literature (see e.g. \citep{zlochin2004model} and \citep{wuthesis}).  This model evolves
iteratively by incorporating information
from some elite solutions occurring in the search history, so as to asymptotically
model the spread of optimal solutions in the search space.  See
the recent Thesis \citep{wuthesis} for more details on MBS algorithms and their mathematical properties. 

An important issue for MBS algorithms is to determine a suitable size for the sampling
in each iteration. A large
sample size makes each iteration unwieldy, however a small sample size may mislead the
underlying search due to the randomness in the sampling. Sample size reflects the {\em iterative complexity} (computational complexity in each iteration). Whether a large sample size is harmful depends on the required \Wurevision{{\em optimization time} (i.e., the total number of iterations required to reach an optimal solution).} 
This article aims to shed a light on this issue by theoretically analyzing the relation between sample size 
and \Wurevision{optimization time} for a CE {\em variant} that includes also some essential
features of the famous \WWu{Max-Min} Ant System ($\mathcal{MMAS}$ \citep{Stuetlze2000}).  To this end, a 
thorough runtime
analysis is needed. 



The theoretical runtime analysis of EAs has gained rapidly growing interest in recent years,
see e.g.
\citep{Stefan1999On}, \citep{He2001}, \citep{Neumann2006},  \citep{witt2006runtime}, \citep{Neumann2009Runtime}, \citep{Frank2011Runtime}, 
\citep{Gutjahr2008Runtime}, \citep{Zhou2007A}, \citep{Zhou2009Runtime}, \citep{Neumann2010A}, \citep{Oliverto2015}, \citep{Sudholt2012}, 
\citep{lissovoi2015}, \citep{chen2014runtime}, and \citep{Sudholt2016Update}.  In the analysis, an oracle-based view of
computation is adopted, i.e., the {\em runtime} of an algorithm is expressed as 
{\em the total number of solutions evaluated before reaching
an optimal solution}. Since the presence of 
randomness, the runtime of an EA is often conveyed in expectation or with high probability. 
Due to the famous No Free Lunch Theorem \citep{Wolpert1997No}, 
the analysis must be problem-specific. The first steps towards this type of analysis were made for the \Wurevision{so-called} (1+1) EA \citep{Stefan1999On}
on some test problems that use {\em pseudo boolean functions} as cost functions, e.g., \textsc{OneMax} \citep{Neumann2009Runtime}, \textsc{LeadingOnes} \citep{Neumann2009Analysis}
and \textsc{BinVar}~\citep{Stefan1999On}. Recent research addresses problems of practical importance, such as the computing a minimum spanning trees (MST) \citep{Neumann2004Randomized},
matroid optimization \citep{Reichel2007Evolutionary}, traveling salesman problem  \citep{Sutton2014Parameterized}, the shortest 
path problem~\citep{lissovoi2015}, the maximum satisfiability problem \citep{Sutton2012SA} and the max-cut problem
\citep{Zhou2015Approximation}.

\Wu{Runtime analysis \Wurevision{generally} considers two cases}: {\em expected runtime analysis} and {\em stochastic runtime analysis}. Expected runtime is the average runtime
of an algorithm on a particular problem, see, e.g., the runtime results of $(1+1)$ EA reported in \citep{Stefan1999On}. Expected runtime reflects the oracle-based 
average performance of an algorithm.  A mature technique for expected runtime
analysis is the so-called drift analysis \citep{He2001}. However, this technique requires that the
algorithm has a finite expected runtime for the underlying problem.  
\Wurevision{By} \citep{Wu2014Asymptotic}, drift analysis is not applicable to the {\em traditional} CE \citep{Rubinstein1999}. 

An algorithm
with a smaller expected runtime need not be more efficient, see 
\citep{Wu2016} for details. In contrast, stochastic runtime
provides a better understanding of the performance of a (randomized) EA. Stochastic
runtime is a runtime result conveyed with an overwhelming probability guarantee (see, e.g., the classic
runtime result of 1-ANT in \citep{Neumann2009Runtime}), where an overwhelming
probability means a probability \Wurevision{tending} to $1$ \Wurevision{superpolynomially} fast in the problem size.
It therefore reflects the efficiency of an algorithm for most cases in the sense
of uncertainty. This article is concerned with stochastic runtime analysis, \Wurevision{
aiming to figure out the relation
between stochastic runtime and \Wurevision{magnitude of} the sample size.} 

\Wurevision{R}untime analysis of CE algorithms has been initiated in
 \citep{Wu2014Asymptotic}, where Wu and Kolonko proved a pioneering stochastic runtime result
for the traditional CE on the standard test problem \textsc{LeadingOnes}. As a
continuation of the study of \citep{Wu2014Asymptotic},
Wu et al \citep{Wu2016} further investigated the stochastic runtime of the traditional
CE on another test problem \textsc{OneMax}. The runtime results reported in
\citep{Wu2014Asymptotic} and \citep{Wu2016}
showed that sample size  plays a crucial role
in efficiently finding an optimal solution. In particular, Wu et al \citep{Wu2016}
showed that if \Wurevision{the problem size $n$
is moderately adapted to the sample size $N$,} then the stochastic runtime of the traditional CE on
\textsc{OneMax} is $O\b{n^{1.5+\frac{4}{3}\epsilon}}$ for arbitrarily small $\epsilon>0$ and
a constant smoothing parameter $\rho>0$,  which beats
the best-known stochastic runtime $O(n^2)$ reported in \citep{Neumann2006} for the classic $1$-ANT algorithm, although $1$-ANT employs a much smaller sample size (i.e., sample size equals one). Moreover,
by imposing upper and lower bounds on the sampling probabilities as was done in
$\mathcal{MMAS}$ \citep{Stuetlze2000}, Wu et al \citep{Wu2016} showed
further  that the stochastic
runtime of the resulting CE can be significantly improved even in a \Wu{very rugged search space.}
\rightnote{I used "So far, so good" to mark how far I read the paper. I probably forgot to delete it. It was not meant to be a Section title. Please choose a better title or do not start a new section\\
	\Wu{I have removed the section title}} 


The present article continues the stochastic runtime analysis of \citep{Wu2016}, but now in combinatorial optimization with a study of CE on the traveling salesman problem (TSP). \Wurevision{We  emphasize the impact of the magnitude of $N$ on 
	the stochastic runtime, put $\rho=1,$ and consider} a
CE variant \Wurevision{resembling} an $\mathcal{MMAS}$ with iteration-best reinforcement under two different
random solution generation mechanisms, namely, a vertex-based random solution generation and
an edge-based random solution generation. 

\Wurevision{
	Stochastic runtime
	for $\mathcal{MMAS}$ with iteration-best reinforcement on simple problems like $\textsc{OneMax}$
	has been studied in \citep{Neumann2010A} and \citep{Sudholt2016Update}. In particular,
	Neumann et al \citep{Neumann2010A} showed that to obtain a stochastically polynomial runtime for \textsc{OneMax}, $N/\rho\in\Omega(\ln n)$ is necessary. We shall not only extend this to TSP for the case of $\rho=1,$ but also prove that $N\in \omega(\ln n)$ is already sufficient to guarantee a stochastically polynomial runtime for simple TSP instances.
}

TSP is a famous $\mathcal{NP}$-complete combinatorial optimization problem. It concerns finding a shortest Hamiltonian cycle on a  weighted complete graph. Existing algorithms exactly
solving TSP
generally have a prohibitive complexity. For instance, the Held-Karp algorithm \citep{Held1962A} solves 
the problem with a complexity of $O(n^22^n).$ 
A well-known polynomial time approximation algorithm for metric TSP is the so-called Christofides algorithm
\citep{Christofides1976Worst}, which finds a solution with a cost at most $3/2$ times the cost of optimal solutions.
As mentioned in \citep{Goodrich2014Algorithm}, this is still the best known approximation algorithm
for the general metric TSP so far. For Euclidean TSP there \Wu{exists} a famous
polynomial-time approximation scheme (PTAS) by Arora, see \citep{Arora98}.
To design a superior approximation algorithm, researchers in recent years tend to study TSP instances with particular structures, see, e.g.,
\citep{Mitchell2010A}.

Due to the prohibitive running time of exact algorithms,  heuristics are frequently employed in practice
so as to efficiently compute an acceptable solution for a TSP problem, e.g., the Lin-Kernighan (LK) algorithm
\citep{Lin1973An}. 
As a popular heuristic, CE  has also been applied in practice to solve TSP instances, see \citep{Boer2014A} and \citep{Rubinstein1999}.  The implementation
there shows that CE can also efficiently compute an acceptable solution.

In view of the high complexity of general TSP, we consider
in our analysis two classes of TSP instances with a particular structure. The first kind of instances
has been used in \citep{Zhou2009Runtime} and \citep{K2012Theoretical} for analyzing
 the expected runtime of some $\mathcal{MMAS}$ variants \Wurevision{with best-so-far reinforcement}. These TSP instances have polynomially
  many objective function values and a unique optimal solution.  \Wurevision{Moreover, on these TSP instances, solutions
 containing more edges from the optimal solution have a smaller cost than those with fewer such edges.} For more details on these instances,
 see Section \ref{sec:RunTime_PCE}. 
 
 For these simple instances,  we \Wurevision{prove in Theorem \ref{theo:G_1_Small_N} that \Wurevision{with a probability $1-e^{-\Omega(n^{\epsilon})},$} the runtime is $O\b{n^{6+\epsilon}}$ with the vertex-based random solution generation,
 and  $O\b{n^{3+\epsilon}\ln n}$ with the edge-based random solution generation,} 
 for any constant $\epsilon\in (0,1)$ and \Wurevision{$N\in \Omega(n^{\epsilon})$}. \Wurevision{For the case of $N\in \omega(\ln n),$
 	we show that the runtimes (resp., $O(n^6N$ and $n^3(\ln n)N$) are even smaller with probability $1-n^{-\omega(1)},$ see Corollary \ref{theo:G1_Perfect}.}
 These results are very close to the known expected runtime $O(n^6+\frac{n\ln n}{\rho})$ for $(1\!+\!1)$ MMAA reported
 in \citep{Zhou2009Runtime}, and the expected runtime $O(n^3\ln n + \frac{n\ln }{\rho})$ for MMAS$^*_{Arb}$ reported
 in \citep{K2012Theoretical} (where $\rho \in (0,1)$ is an evaporation rate), \Wurevision{if $N\in \omega(\ln n)$ is suitably small.}  
 But
 they give the stronger guarantee of achieving the optimal solution in the respective 
 runtime with an overwhelming probability. \Wurevision{Moreover, \SecondRevision{we show} 
 	a stochastically exponential runtime 
 	for a suitable choice of $N\in O(\ln n)$, see Theorem \ref{theo:G_1_VerySmall_N}. This generalizes the 
 	finding in \citep{Neumann2010A} for \textsc{OneMax} to TSP instances.
 	Therefore,
 	$N\in \Omega(\ln n)$ is necessary, and $N\in \omega(\ln n)$ 
 	is sufficient for a stochastically polynomial runtime for
 	simple TSP instances.}


 We also inspect more complex instances with $n$ vertices positioned on an $m\times m$ grid, and \WWu{the Euclidean distance} as distance function.
 These instances have been employed in \citep{Sutton2012A} and \citep{Sutton2014Parameterized} for
 analyzing the expected runtime of $(\mu\!+\!\lambda)$ EA and randomized local search (RLS). When the $n$ vertices
 span a convex polygon without vertices in the interior of the polygon (so they
 are the corners of that polygon),
 we prove a stochastic runtime of \Wurevision{$O(n^{4}m^{5+\epsilon})$} for the vertex-based random solution generation,
 and a stochastic runtime of \Wurevision{$O(n^{3}m^{5+\epsilon})$} for the edge-based random solution generation,
 see Theorem \ref{theo:TSP_Convex_Case} for details. \Wurevision{Similarly,
 	the $\epsilon$ in the stochastic runtimes can be removed by slightly
 	decreasing the probability guarantee, see Corollary \ref{theo:Corr2}.}
 When the vertices span a convex polygon with $k \in O(1)$ vertices in the interior, we show
 a stochastic runtime of \Wu{$O(n^{4}m^{5+\epsilon}+n^{6k-1}m^{\epsilon})$} with the vertex-based random solution
 generation, and a stochastic runtime of \Wu{$O(n^{3}m^{5+\epsilon}+n^{3k}m^{\epsilon})$} with the edge-based 
 random solution generation, see
 Theorem \ref{theo:grid_2} for details. These runtimes
 are better than the expected runtime for the so-called $(\mu\!+\!\lambda)$ EA and RLS reported in the recent paper
 \citep{Sutton2014Parameterized}.

The remainder of this paper is arranged as follows. 
Section \ref{sec:algorithm} defines the traditional CE and related algorithms, 
Section \ref{sec:TSP} defines the traveling salesman problem and provides more details of the
used CE variants, \Wurevision{Section \ref{sec:RNG_Property} shows some important facts on the two random solution generation methods}, and
Section \ref{sec:RunTime_PCE} reports the stochastic runtime results on the TSP instances. A short conclusion
and \WWu{suggestions for future work are given in Section \ref{sec:conclusion}.}

\section*{Notations for runtime}
\label{sec:notations}
Our analysis 
employs some commonly used notations from complexity theory. We use $O\b{f(n)}$ to denote the class
of functions which are {\em bounded from above} by the function $f(n)$, \Wurevision{i.e., those functions $g(n)$
with $g(n) \le c\cdot f(n)$ for large enough $n$ and some constant $c\ge 0$ not depending on $n.$} Similarly,
$\Omega\b{f(n)}$ is the class of functions that are {\em bounded from below} by $f(n)$, i.e., 
for any $g(n)\in \Omega\b{f(n)}$ there exists a constant $c>0$ not depending on $n$ such
that $g(n)\ge c\cdot f(n)$ \Wurevision{for large enough $n.$} Class $\Theta\b{f(n)}$ is the intersection of $\Omega\b{f(n)}$
and $O\b{f(n)}.$ Class $o\b{f(n)}$ is the class of functions $g(n)$ with 
$g(n)/f(n)\to 0$ as $n \to \infty,$ and class $\omega\b{f(n)}$ is the class of
functions $g(n)$ with $g(n)/f(n) \to +\infty$ as $n \to \infty.$ Obviously,
$o\b{f(n)} \subset O\b{f(n)}$ and $\omega\b{f(n)} \subset \Omega\b{f(n)}.$ 

\section{The general cross entropy algorithm and related algorithms} 
\label{sec:algorithm}

We now \Wu{introduce} the traditional CE algorithm. 
\Wu{The CE variant we will analyze inherits the framework
	of this traditional version.} 
To compare our results with those in
the literature, we shall give also details about some related algorithms.

\subsection{The \WWu{traditional} cross entropy algorithm}
\label{sec:algorithm_CE}


Algorithm
\ref{alg:ce} lists the traditional CE that was proposed in \citep{Rubinstein1999}, adapted to an abstract notion of combinatorial optimization problems.
The algorithm assumes a combinatorial minimization problem $(S, f )$, \Wu{where
$S$ is a {\em finite} search space \WWu{of ``feasible" solutions
	and $f$ is the {\em cost} function. Every} feasible
solution $s\in S$ is composed \WWu{of elements from} a fixed finite set $\A$, the ground set of the problem, i.e., we assume $S\subseteq \A^n$ for some
integer $n\in \N$.}   
\rightnote{Now I understand the model much better}  
\rightnote{The explanation of $\A^n$ is misleading. Later, $p_{ij}$ is the probability to continue with edge $\{i,j\}$ from vertex $i$ in a cycle. So factor $i$ of  $\A^{n}$ contains only probabilities for the pairs $(i,j)$, $i$ fixed.  You do not have edge probabilities
	\Wu{I added a little more details in the next paragraph and will give full details
		in the next section.}} 
Furthermore there is a product distribution on the product space \WWu{$\A^n$}  that induces a distribution on
$S\subseteq \A^n.$ 
The
distribution on $\A^n$ can usually be represented as a vector (or matrix) of \WWu{real-valued probabilities}.
The convex combination of the two distributions in Step 6 of Algorithm \ref{alg:ce}
then corresponds to a convex combination of the two vectors (or matrices).

\Wu{Specific to the TSP, the ground set 
$\A$ can be the set of nodes or edges, $n$ is the number of nodes, and
a feasible solution \WWu{is a sequence of elements from $\A$ that forms a Hamiltonian cycle.}}  \Wu{The product distribution
for the TSP is represented as an $n\!\times\! n$ matrix.
\rightnote{\WWu{Let the following as a new paragraph as you suggested!}}

 When we consider the
set of nodes as our ground set $\A$, each row $i$ of the matrix is a marginal
distribution that specifies \WWu{choice probabilities
for {\em all} nodes following the current node $i$.} A random Hamiltonian \WWu{cycle} is 
sequentially constructed from the product distribution  by \WWu{allowing
only nodes} not yet visited as continuations in each step, see Algorithm \ref{alg:TSP_RNG_Vertex} for more details. 

When we consider the
set of edges as $\A,$ marginals of the product 
distribution \WWu{will be represented 
by} the same $n\!\times\! n$ matrix where the sum of the $(i,j)$-th 
and $(j,i)$-th entries \WWu{reflects
the probability that the edge $\{i,j\}$ occurs in a random solution.}
A random Hamiltonian \WWu{cycle is still constructed  sequentially
	and only edges leading to a feasible solution are taken in each step,} see Algorithm \ref{alg:TSP_RNG_Edge} for details.}
\begin{algorithm}[!htb]
\caption{The general Cross-Entropy algorithm}
\label{alg:ce}
\begin{algorithmic}[1]
\REQUIRE ~~\\
  an {\em initial distribution} $\PPi_0$ on the solution space, a fixed {\em smoothing
  parameter} $\rho\in (0,1],$ a {\em sample size} $N\in \N_+,$ an {\em elite size} $M \in \N_{+}$ with
  $M \le N$
\STATE $t=0;$
\LOOP 
\STATE independently generate $N$  random solutions $\X_t^{(1)}, \ldots, \X_t^{(N)}$ \WWu{with \Wu{the current distribution}} $\PPi_t;$
\STATE sort these $N$ \WWu{solutions in non-decreasing} order as $f\b{\X_t^{[1]}}\le \cdots\le f\b{\X_t^{[N]}}$ according to
the cost function $f$;
\STATE learn an empirical distribution $\W_t$ from the $M$ best solutions $\X_t^{[1]}, \ldots, \X_t^{[M]}$;
\STATE set $\PPi_{t+1}=(1-\rho)\PPi_t + \rho \W_t;$
\STATE $t=t+1;$
\ENDLOOP
\end{algorithmic}
\end{algorithm}

Traditionally, CE sets a small elite ratio $\alpha\in (0,1)$ and uses the best $\fl{\alpha \cdot N}$
 solutions in Step 5 to build the empirical distribution $\W_t$. Here, we use the elite size $M$
instead. This does not intrinsically change the original algorithm. Steps 3 and 5 depend on the detailed definition of the underlying problem.
We shall give details to them in Subsection \ref{sec:CE_details}.

Step $6$  of Algorithm \ref{alg:ce} plays a crucial role in the different theoretical analyses of the algorithm,
see, e.g., \citep{Costa2007Convergence}, \citep{Wu2014Asymptotic}, \citep{wu2014absorption}, 
\citep{wuthesis}, \citep{Wu2016}. The occurrence of good solutions are probabilistically enforced by
incorporating the new information $\W_t$ into $\PPi_{t+1}.$ This idea, somehow, coincides with the reinforcement learning in \citep{thomas1997}.
The smoothing parameter $\rho$ reflects the relative importance
of the new information $\W_t$ in the next sampling. It balances global exploration and local exploitation
to a certain degree.
A larger $\rho$ makes the algorithm concentrate more on the particular area spanned by the elite solutions
$\X_t^{[1]}, \ldots, \X_t^{[M]}$, while
a smaller $\rho$ gives more opportunities to solutions outside that area.

However, balancing global exploration and local exploitation through tuning $\rho$ is ultimately
limited.  Wu and Kolonko \citep{Wu2014Asymptotic}
proved that the famous
``genetic drift" \citep{Asoh1994} phenomenon also happens in this algorithmic scheme, i.e., the sampling
(Step 3) 
eventually freezes at a single solution and that solution \Wu{needs} not to be optimal.
This means that the algorithm gradually loses the power of global exploration.

As a compensation for global exploration, Wu et al \citep{Wu2016} proved that a moderately large sample size $N$ might be helpful. The results there showed that a moderately large $N$ configured with a large $\rho$ (e.g., $\rho=1$)
can make the algorithm very efficient.
Although a large $N$  introduces a high computational burden in each iteration, the total
number of iterations required for getting an optimal solution is considerably reduced. 

Wu et al \citep{Wu2016} also indicated another way to compensate the global exploration, i.e.,
imposing a lower bound \WWu{$\pi_{\min} \in (0,1)$} and
an upper bound \WWu{$\pi_{\max}\in (0,1)$} on the sampling distributions in each iteration. This idea
is originated from $\mathcal{MMAS}$ \citep{Stuetlze2000}.
In each iteration $t,$ {\em after} applying Step 6, the entries of distribution $\PPi_{t+1}$ \WWu{that are} out of the range
$[\pi_{\min}, \pi_{\max}]$ are reset to that range by assigning to them the nearest bounds,
see \eqref{eq:adjustment} in Section \ref{sec:TSP} for more details.
Wu et al \citep{Wu2016} \WWu{have proved} that this can make
CE more efficient even in the case of \Wu{a rugged search space. }

To follow these theoretical suggestions made
in \citep{Wu2016}, we shall in our stochastic runtime analysis use a CE 
that modifies the traditional CE (Algorithm \ref{alg:ce}) accordingly. 
We shall see that these modifications make the CE very efficient for the considered \WWu{TSP instances.}

\subsection{Related evolutionary algorithms}
\label{sec:algorithm_related_EA}


Related evolutionary algorithms for TSP whose runtime has been extensively studied are RLS \citep{Neumann2004Randomized}, $(\mu + \lambda)$ EA \citep{Sutton2014Parameterized}, 
and those theoretical abstractions of $\mathcal{MMAS}$ \citep{Stuetlze2000}
including 
MMAS$^*_{bs}$ \citep{Gutjahr2008Runtime}, 
(1+1) MMAA \citep{Zhou2009Runtime}. We now give algorithmic details of them.
\WWu{In order to facilitate the comparison, their runtimes for
	TSP instances will be discussed in Section \ref{sec:RunTime_PCE}.}

$(\mu \!+\!\lambda)$ EA is an extension of the famous $(1\!+\!1)$ EA \citep{Stefan1999On}. $(\mu \!+\!\lambda)$ EA randomly chooses
\Wurevision{$\mu$} solutions as the initial population. In each iteration, $(\mu \!+\!\lambda)$ EA  
randomly chooses $\lambda$ parents from current population, then produces $\lambda$ children by
applying randomized mutation to each of the selected parents, and forms the next population by taking
the best $\mu$ solutions from these $\mu\!+\!\lambda$ solutions in the end of current iteration.
The expected runtime of $(\mu \!+\!\lambda)$ EA on TSP instances is studied in \citep{Sutton2014Parameterized},
where Sutton et al uses a Poisson distribution to determine the number of randomized mutations (2-opt move or
jump operation) should be taken by a selected parent in each iteration.

RLS is a local search technique \citep{Pirlot1996General}.  It employs a randomized
neighborhood. In each iteration, it randomly chooses a  number of components of the
best solution found so far and then changes these components. The expected runtime of
RLS for TSP instances is also studied in \citep{Sutton2014Parameterized}, where 
\rightnote{the referee for the first paper suggested to replace Poisson by Bernoulli\\
\Wu{It is not the Poisson trials, it is Poisson distribution here.}} 
the neighborhood is taken to be a $k$-exchange neighborhood with $k$ randomly determined
by a Poisson distribution.

$(1\!+\!1)$ MMAA \WWu{is a simplified version} of the famous $\mathcal{MMAS}$ \citep{Stuetlze2000}, where
the sample size is \Wurevision{set to $1$ and} pheromones are updated only with the best solution found
so far ({\em best-so-far reinforcement}) in each iteration. In each iteration of $(1\!+\!1)$ MMAA,
the ant which constructed the best solution found so far deposits an amount  $\pi_{\max}$ of pheromones
on the traversed edges, and an amount $\pi_{\min}$ of pheromones on the non-traversed edges, and
the pheromones are updated by linearly combining the old and these newly added pheromones as in Algorithm
\ref{alg:ce}. The expected runtime of $(1\!+\!1)$ MMAA on simple TSP instances is studied in 
\citep{Zhou2009Runtime}. The expected runtime of its variant MMAS$^*_{Arb}$ on simple TSP instances is studied
in \citep{K2012Theoretical}.

\section{The traveling salesman problem and \Wu{details of the CE variant}}
\label{sec:TSP}


Now, we formally define TSP, and give more details of the CE \Wu{variant} we will analyze.
\subsection{The traveling salesman problem}
\label{sec:TSP_Definition}
 We consider an undirected graph $G=(V,E)$ with 
vertex set $V=\{1,\ldots,n\}$ and edge set $E=\big\{\{i,j\}\ |\ i\in V, j\in V, i\ne j\big\}.$
\rightnote{\Wu{I changed the ',' to '.'}\\
	\WWu{Probably, we should use sequence instead of set, because this representation is more consistent with our formalization of the problem!}}
A {\em Hamiltonian cycle} is \WWu{a sequence} $\{\{i_l, i_{l+1}\}\ |\ l=1,\ldots,n\}$ of edges  such that
\begin{itemize}
\item[a)] $i_1=i_{n+1};$
\item[b)] $(i_1, \ldots, i_n)$ is a permutation of $\{1,2,\ldots,n\}.$
\end{itemize}
\Wu{This definition actually considers $E$ as the ground set $\A.$ As mentioned above, we can also put $\A=V$
	and represent Hamiltonian cycles in a more
	compact way as permutations of $V.$}
Note that
a Hamiltonian cycle corresponds to $n$ different permutations, whereas a permutation
corresponds to a unique Hamiltonian cycle. \Wu{However, the two representations are intrinsically
\WWu{the same.} We shall use them interchangeably in the sequel.}
To facilitate our discussion, we shall refer to a Hamiltonian cycle by just referring to one
of the $n$ corresponding permutations, \Wu{and \WWu{denote by $S$} the set of all possible 
	permutations}. We employ the convention that two permutations
are said to be same iff they form the same underlying Hamiltonian cycle.
The notation $\{k,l\} \in s$ shall mean that the edge $\{k,l\}$ {\em belongs to} the
underlying Hamiltonian cycle of the
solution (permutation) $s.$ 

Once a distance function $d: E \mapsto \R_{+}$ is given, the {\em (total
traveling) cost} $f(s)$ of a feasible solution $s=(i_1, i_2, \ldots, i_n)\in S$ is then calculated by
\begin{equation}
\label{eq:TSP_Evaluation}
f(s):= \sum_{j=1}^{n-1} d(i_j, i_{j+1}) + d(i_n, i_1).
\end{equation}
We denote by $S^*\subseteq S$ the set of feasible solutions (Hamiltonian cycles) that 
minimize the  cost \eqref{eq:TSP_Evaluation}. 



\subsection{Details of the CE variant}
\label{sec:CE_details}
The CE variant we consider in the analysis completely inherits the structure of Algorithm \ref{alg:ce}, and additionally employs a component from $\mathcal{MMAS}.$

We now formalize the sampling distribution,
and define Steps 3 and 5 in more detail. \Wu{As mentioned, we represent a
	sampling distribution (a product distribution on $\A^n$) for the TSP}
\WWu{by a} matrix $\PPi=(\pi_{i,j})_{n\times n},$ such that 
\begin{itemize}
\item[a)] $\sum_{j=1}^{n} \pi_{i,j}=1,$ for all $i=1, \ldots,n,$
\item[b)] $\pi_{i,i}=0$ for all $i=1,\ldots,n,$
\item[c)] $\pi_{i,j}=\pi_{j,i}$ for each edge $\{i,j\}\in E.$
\end{itemize}
\rightnote{I would say that  is the probability that a Hamilton cycle continues with vertex $j$ when it is in vertex $i$. The edge probability is $\pi_{ij}+\pi_{ji}$ \Wu{You are right!}} 
For each edge $\{i,j\}\in E,$ $\pi_{i,j}$ \WWu{reflects the probability} that  \Wu{a Hamilton cycle continues with vertex $j$ when it is in vertex $i$}. In the sequel, we write
the sampling distribution $\PPi_t$ in iteration $t$ as $(\pi_{i,j}^t)_{n\times n},$ where the
superscript $t$ of $\pi_{i,j}^t$ indicates the iteration. The {\em initial distribution} $\PPi_0=(\pi_{i,j}^0)_{n\times n}$
is, without loss of generality, set to be the uniform distribution, i.e., $\pi_{i,j}^0=\pi_{j,i}^0=\frac{1}{n-1}$ for all
edges $\{i,j\}\in E.$

We shall consider two random solution generation methods, a {\em vertex-based random solution generation}
and an {\em edge-based random solution generation}. 
\Wu{Algorithm \ref{alg:TSP_RNG_Vertex} lists the vertex-based random solution generation  method.} 
\Wu{This method uses $V$ as the ground set $\A$. A product
	distribution of $\A^n$ is therefore represented as a matrix $\PPi=(\pi_{i,j})_{n\times n}$ \WWu{ satisfying
	a)-c) above,} i.e., each row of 
	$\PPi$ represents a sampling distribution on $\A=V.$ Directly sampling
	from $\PPi$ may produce infeasible solutions from $\A^n-S.$ To 
	avoid that, Algorithm \ref{alg:TSP_RNG_Vertex} starts with a 
	randomly fixed initial node, and then sequentially extends a partial solution with an unvisited vertex until a complete permutation is obtained.}
This method is efficient and rather popular in practice, see, e.g., \citep{Boer2014A} and \citep{dorigobook}. Here, ``$s+(v)$" means that appends a vertex
 $v$ to the
end of a partial solution $s.$ 
\begin{algorithm}[!htb]
\caption{Vertex-based random solution generation}
\label{alg:TSP_RNG_Vertex}
\begin{algorithmic}[1]
\REQUIRE ~~\\
    a distribution $\PPi=\b{\pi_{i,j}}_{n\times n}$
\ENSURE ~~\\
    a permutation of $1,2,\ldots,n$
\STATE $s=\emptyset,$ and $V_{univisted}=V;$ 
\STATE randomly select $v$ from $V,$ $s = s + (v),$ and $V_{unvisited}=V_{unvisited}-\{v\};$
\WHILE{$(|V_{unvisited} \neq \emptyset|)$}
\STATE select a random vertex $v'$ from $V_{unvisited}$ with a probability
\begin{equation}
\label{eq:TSP_continuation_select_prob}
\P[v'\ |\ s]=\frac{\pi_{{v, v'}}}{\sum_{k\in V_{unvisited}} \pi_{{v,k}}};
\end{equation}
\STATE set $s\! =\! s\!+\!(v'),$ $V_{unvisited}=V_{unvisited}-\{v'\};$
\STATE  $v=v';$
\ENDWHILE
\RETURN $s;$
\end{algorithmic}
\end{algorithm}

The edge-based random solution generation  is listed in Algorithm \ref{alg:TSP_RNG_Edge}.
The idea is from 
\citep{K2012Theoretical}.
This method considers \WWu{edge set} $E$ as \Wu{the ground set $\A$}. 
\Wu{A feasible solution is then a sequence of edges \WWu{that form} a Hamiltonian \WWu{cycle,} i.e. $S\subseteq E^n$. 
	To unify the notation of feasible solutions, Algorithm \ref{alg:TSP_RNG_Edge} \WWu{translates} its outcomes into permutations. As the actual 
	ground set is $E,$ a product distribution is an $n\times \frac{n(n-1)}{2}$ matrix such that each row is a marginal specifying a sampling distribution on $E.$
	Algorithm \ref{alg:TSP_RNG_Edge} only considers those with identical marginals, a product
	distribution can be therefore fully characterized by one of its marginals and is therefore again \WWu{represented
	by an $n\times n$ matrix $\PPi=(\pi_{i,j})_{n\times n}$ as above. }
	An edge $\{i,j\}\in E$ is then sampled from $\PPi$ with probability $(\pi_{i,j}+\pi_{j,i})/\sum_{k=1}^{n}\sum_{l=1}^{n} \pi_{k,l}=2\pi_{i,j}/n$
	since each row of $\PPi$ sums up to $1$. A random \WWu{sequence $\in E^n$} is generated
	by independently sampling from $\PPi$ $n$ times. To avoid infeasible solutions, Algorithm \ref{alg:TSP_RNG_Edge} \WWu{ considers  in every sampling only edges} that are {\em admissible} by the edges selected before.}  Given a set $\mathcal{B}$ of edges such that the subgraph  $(V, \mathcal{B})$ does neither
contain a cycle nor a vertex of degree $\ge 3,$ an edge \WWu{$e' \in E$} is said to be  admissible by $\mathcal{B}$
if and only if the subgraph \WWu{$(V, \mathcal{B}\cup\{e'\})$} still does neither contain a cycle nor a vertex of degree
$\ge 3.$ 
We denote by $B_{admissible}$ the set of edges $\notin \mathcal{B}$ that are admissible by $\mathcal{B}.$
\begin{algorithm}[!htb]
\caption{Edge-based random solution generation}
\label{alg:TSP_RNG_Edge}
\begin{algorithmic}[1]
\REQUIRE ~~\\
    a distribution $\PPi=\b{\pi_{i,j}}_{n\times n}$
\ENSURE ~~\\
    a permutation of $1,2,\ldots,n$
\STATE $\mathcal{B}=\emptyset, B_{admissible}=E;$
\WHILE{$(|\mathcal{B}|\le n-1)$}
\STATE select an edge \Wu{$\{i,j\}$} from $B_{admissible}$ with a probability
\begin{equation}
\label{eq:TSP_edge_select_prob}
\P[e\ |\ s]=\frac{\pi_{i,j}+\pi_{j,i}}{\sum_{\{k,l\}\in B_{admissible}} \Wu{\pi_{k,l}+\pi_{l,k}}};
\end{equation}
\STATE set \Wu{$\mathcal{B}\! =\! \mathcal{B}\!\cup\!\{\{i,j\}\};$}
\STATE update $B_{admissible};$
\ENDWHILE
\STATE let $s=(1,i_2, i_3,\ldots,i_n)$ with $\{1,i_2\}, \{i_j, i_{j+1}\}\in \mathcal{B}$ for \Wurevision{$j\!=\!2,\ldots,n\!-\!1;$}
\RETURN $s;$
\end{algorithmic}
\end{algorithm}

The $N$ random solutions $\X_t^{(1)}, \ldots, \X_t^{(N)}$ in iteration $t$ are then generated by $N$ runs of Algorithm \ref{alg:TSP_RNG_Vertex} or Algorithm \ref{alg:TSP_RNG_Edge} with \WWu{\em  the current distribution} $\PPi_t=(\pi_{i,j}^t)_{n\times n}$. The empirical distribution
$\W_t=(w_{i,j}^t)_{n\times n}$ is then calculated from the $M$ elite solutions by setting 
\begin{equation}
\label{eq:empirical_dist}
w_{i,j}^t=\frac{\sum_{k=1}^{M} \1_{\{e'\in E\ | \ e' \in \X_t^{[k]} \}}(\{i,j\})}{M},
\end{equation}
\rightnote{I think in (3) above it should be $\pi_{kl}+\pi_{lk}$ in the denominator, since an edge $\{i,j\}$ occurs only once in the sum\\
\Wu{Thank you for pointing out this mistake!}}
\WWu{where} $\1_A(\cdot)$ is the indicator function of set \Wurevision{$A = \{e'\in E\mid e'\in \X_t^{[k]}\}$} for each $\{i,j\}\in E.$ The next 
distribution $\PPi_{t+1}=(\pi_{i,j}^{t+1})_{n\times n}$ is therefore obtained as
\begin{equation}
\label{eq:pi_t+1_details}
\pi_{i,j}^{t+1}=(1-\rho)\pi_{i,j}^t+\rho w_{i,j}^t
\end{equation}
for each $\{i,j\}\in E.$

We \Wu{ continue with} the suggestions made in \citep{Wu2016}. In the CE variant, we shall use a moderately large $N$ and a large $\rho=1$.
To fully use the best elite solutions, we take $M=1.$
To prevent premature convergence (i.e., a possible stagnation at a non-optimal solution), we \Wurevision{employ a feature} from $\mathcal{MMAS}$ 
\citep{Stuetlze2000}, called  {\em max-min calibration}, in the construction of $\PPi_{t+1},$. We choose a lower bound $\pi_{\min}\in (0,1)$ and an upper bound $\pi_{\max}\in (0,1),$
and, after applying \eqref{eq:pi_t+1_details}, adjust $\PPi_{t+1}$ 
by 
\begin{equation}
\label{eq:adjustment}
\pi_{i,j}^{t+1}=
\begin{cases}
\pi_{\min} &\text{if } \pi_{i,j}^{t+1}<\pi_{\min},\\
\pi_{i,j}^{t+1} &\text{if } \pi_{i,j}^{t+1}\in [\pi_{\min}, \pi_{\max}],\\
\pi_{\max} &\text{if } \pi_{i,j}^{t+1} >\pi_{\max},
\end{cases}
\end{equation}
for any edge $\{i,j\}\in E$. Note that
the max-min calibration is the {\em only} step that does not occur
in the general CE (i.e., Algorithm \ref{alg:ce}).

This setting turns CE into \Wurevision{an $\mathcal{MMAS}$ with {\em iteration-best reinforcement}, i.e., only the iteration-best solution $\X_t^{[1]}$ is allowed to change
the `pheromones' \Wu{$\PPi_t$}.} \Wurevision{St{\"utzle} and Hoos \citep{Stuetlze2000} indicated 
in an empirical study that the practical performance of iteration-best reinforcement is comparable
to best-so-far reinforcement for TSP instances. Thus,
it should also be worthwhile to compare the theoretical runtime
of iteration-best reinforcement with the known expected runtimes of best-so-far
reinforcement for TSP instances presented in, e.g., \citep{Zhou2009Runtime} and \citep{K2012Theoretical}.} 


\Wurevision{\section{Properties of the random solution generation methods}
	\label{sec:RNG_Property}}

\Wurevision{Before we start with our runtime analysis, we shall discuss some relevant 
	properties of the two random solution generation methods, which concern the
	probability of producing a $k$-exchange move of the iteration-best solution in the next
	sampling.

	Formally, a {\em $k$-exchange move} on a Hamiltonian cycle is an operation
	that removes $k$ edges from the cycle and adds $k$ new edges to obtain again a cycle. A {\em $k$-opt move} is a $k$-exchange move reducing the total travel cost.	
	Figure \ref{fig:2exchange} shows an example of 
	a 2-exchange move, in which edges $\{i,j\},\{k,l\}$ are removed, and edges $\{i,l\}, \{k,j\}$
	are added.Figure \ref{fig:3exchange} shows an example
	of a 3-exchange move.
	\begin{figure}[!htb]
		\begin{subfigure}{0.47\textwidth}
			\includegraphics[scale=0.5]{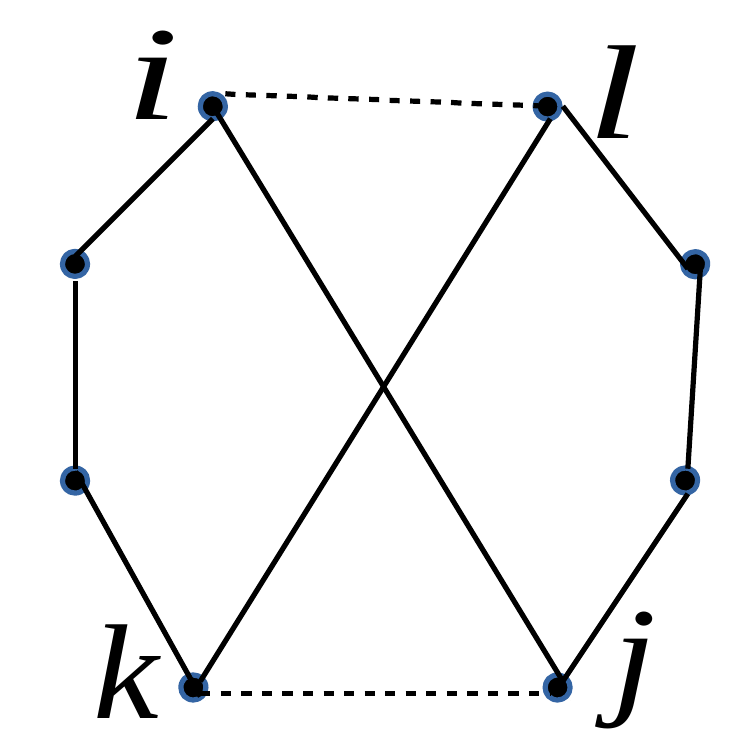}
			\caption{A 2-exchange move}
			\label{fig:2exchange}
		\end{subfigure}
		\begin{subfigure}{0.47\textwidth}
			\includegraphics[scale=0.5]{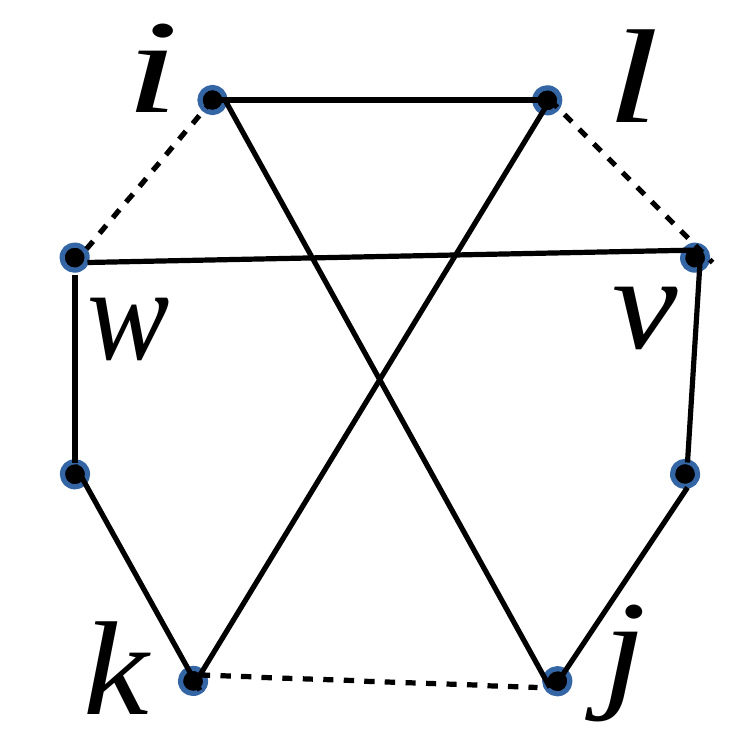}
			\caption{A 3-exchange move}
			\label{fig:3exchange}
		\end{subfigure}
		\caption{Examples for edge exchange moves}
		\label{fig:edgeexchange}
	\end{figure}
	
	In our analysis, we shall consider only iteration-best reinforcement with
	$\rho=1$ and the
	max-min calibration \eqref{eq:adjustment}. The empirical distribution $\W_t=(w_{i,j}^t)_{n\times n}$ for each iteration $t\in N$ in this particular case therefore satisfies
	\begin{equation}
	\label{eq:G1_pi_t}
	\Wu{\pi_{i,j}^{t+1}=\pi_{j,i}^{t+1}=}
	\begin{cases}
	\Wu{\min\{1,\pi_{\max}\}=\pi_{\max}} &\text{if \WWu{edge} } \{i,j\} \in \X_t^{[1]},\\
	\Wu{\max\{0,\pi_{\min}\}=\pi_{\min}} &\text{otherwise},
	\end{cases}
	\end{equation}
	\Wu{for every edge $\{i,j\}\in E$ and iteration $t \in \N.$ Furthermore, $\PPi_{t\!+\!1}=\W_t$.}
	
	Since $\PPi_{t\!+\!1}$ is biased towards the iteration-best solution $\X_t^{[1]},$ $k$-exchanges of $\X_t^{[1]}$ with a large $k$ are unlikely to happen among the $N$ draws from $\PPi_{t\!+\!1}$ by 
	either of the two generation methods. Thus, an optimal solution is more likely to be reached by a sequence of repeatedly
	$k$-exchange moves with small $k$ from iteration-best solutions. 
	Therefore, it is necessary to estimate the probabilities of producing a $k$-exchange of $\X_t^{[1]}$ 
in the two generation methods, especially for the case of small $k.$ 
}

	\Wurevision{
	
\subsection{Probabilities of producing $k$-exchanges in the vertex-based random solution generation}
	
	The probability of producing $k$-exchanges with $k=2,3$ in the vertex-based random solution generation has been studied in Zhou \citep{Zhou2009Runtime}. With $\pi_{\min}=\frac{1}{n^2}$
	and $\pi_{\max}=1-\frac{1}{n},$ Zhou \citep{Zhou2009Runtime} proved for $(1\!+\!1)$ MMAA that
	with a probability of $\Omega\b{1/n^5},$ Algorithm \ref{alg:TSP_RNG_Vertex} produces a random 
	solution having more edges from $s^*$ than  $x_t^*$  (the best solution found so far) provided that $x_t^*$ is not optimal.
	Zhou \citep{Zhou2009Runtime}
	actually showed that if $x_t^*\ne s^*,$  then there exists either a 2-opt move or a 3-opt move for $x_t^*,$
	and Algorithm \ref{alg:TSP_RNG_Vertex} produces 
	an arbitrary 2-exchange of $x_t^*$  with a probability of $\Omega\b{1/n^3}$, and an arbitrary 3-exchange
	of $x_t^*$ with a probability of $\Omega\b{1/n^5}.$
	
	Although we use $\pi_{\min}=\frac{1}{n(n-2)}$ and
	consider iteration-best reinforcement, a similar 
	result holds in our case. Claim \ref{theo:Lemma_RNG_Vertex}
	below gives a \SecondRevision{lower bound on the} probability of producing a
	$k$-exchange move of the iteration-best solution in the next round with the vertex-based random solution generation. 
	\begin{claim}\label{theo:Lemma_RNG_Vertex}
		Let $M\!=\!1, \rho=1,$ and consider a $k$-exchange move
		of $\X_t^{[1]}$ for some integer $k=2,3,\ldots, n.$ Then, Algorithm 
		\ref{alg:TSP_RNG_Vertex} produces the given $k$-exchange move 
		with a probability $\Omega(1/n^{2k-1})$ in every
		of the $N$ draws in iteration $t+1.$
	\end{claim}   
	\begin{proof}
		
		Recall that in Algorithm \ref{alg:TSP_RNG_Vertex}, the probability
		\eqref{eq:TSP_continuation_select_prob} to select a continuing edge $\{i,j\}$ is always bounded from below by $\pi^t_{i,j}$ \WWu{(or, equivalently, $\pi_{j,i}^t$)}
		for each iteration $t\in \N,$ since each row of $\PPi_t$ sums up to $1.$
		Given a $k$-exchange move of $\X_t^{[1]},$ one possibility to generate it from 
		$\PPi_{t\!+\!1}=\W_t$ by Algorithm \ref{alg:TSP_RNG_Vertex} is that one of \SecondRevision{the
		new edges} is added in the last step. This happens with a probability at least
		$
		\frac{1}{n}\cdot\big[\frac{1}{n(n-2)}\big]^{k-1}\cdot \big(1-\frac{1}{n}\big)^{n-k}\ge \frac{1}{e\cdot n^{2k-1}},
		$
		where \WWu{$e\approx 2.71828$ is Euler's number,} $\frac{1}{n}$ represents the probability to select the starting vertex, $\frac{1}{n(n-2)}$
		is the common lower bound of the probability to select the remaining $k-1$ new edges, and $1-\frac{1}{n}$ is
		the common lower bound of the probability to select one of the remaining $n-k$ edges
		\Wu{from $\X_t^{[1]}$}. 
		
	\end{proof}
	
	Because of Claim \ref{theo:Lemma_RNG_Vertex}, every $2$-exchange of $\X_t^{[1]}$ is produced 
	from $\PPi_{t+1}$ by 
	Algorithm \ref{alg:TSP_RNG_Vertex} with a probability $\Omega(1/n^{3}),$ and every 
	$3$-exchange is produced  by Algorithm \ref{alg:TSP_RNG_Vertex} with a probability
	$\Omega(1/n^5).$ Note that for any $k=2,3, \ldots,$ if a $k$-opt move of 
	$\X_t^{[1]}$ occurs among the $N$ draws in the next sampling,
	then $f(\X_{t+1}^{[1]})<f(\X_t^{[1]})$ must hold. Thus, if we take a moderately large sample size, say $N=\Theta(n^{5+\epsilon})$
	for some $\epsilon>0,$ with a probability $1-(1-\Omega(n^{-5}))^{\Omega(5+\epsilon)}=1-e^{-\Omega(n^{\epsilon})},$ $f(\X_{t+1}^{[1]})<f(\X_t^{[1]})$ will hold, provided that there still exists a $2$-opt or $3$-opt move
	from $\X_t^{[1]}.$
	
	\begin{claim}\label{theo:Vertex_nonincreasing}
		\SecondRevision{Suppose that $M\!=\!1,\rho\!=\!1.$  Then, for iteration
			$t+1,$ the probability that Algorithm \ref{alg:TSP_RNG_Vertex}
			produces a solution with a cost not larger than $\X_t^{[1]}$ in one
			application is in $\Omega(1).$}
	\end{claim}
	\begin{proof}
		Observe that \SecondRevision{the probability that $\X_t^{[1]}$ is reproduced in one application of Algorithm 
		\ref{alg:TSP_RNG_Vertex}} is larger than \SecondRevision{$(1-1/n)^{n-1}\in \Omega(1),$} which implies
		that the cost of the generated random solution  is not larger than $f(\X_t^{[1]}).$ 
	\end{proof}
	Note that \SecondRevision{if $\X_t^{[1]}$ is reproduced at least once among
	the $N$ draws in the next sampling, then} $f(\X_{t+1})^{[1]}\le f(\X_t^{[1]}).$
	Thus, if the sample size $N\in \Omega(\ln n),$ then $f(\X_{t+1}^{[1]})\le f(\X_t^{[1]})$  with a probability $1-(1-\Omega(1))^N=1-O(1/n).$ Particularly, when $N\in \Omega(n^{\epsilon})$
	for some $\epsilon>0,$ $f(\X_{t+1}^{[1]})\le f(\X_t^{[1]})$ with an overwhelming probability $1-e^{-\Omega(n^{\epsilon})}.$
}

\Wurevision{
	
	\subsection{Probabilities of producing $k$-exchanges in the edge-based random solution generation}
	
	The behavior of the edge-based random solution generation is comprehensively studied in \citep{K2012Theoretical}.
	K{\"o}tzing et al \citep{K2012Theoretical} proved for MMAS$^*_{Arb}$
	and a constant $k\in O(1)$ that, with a probability of
	$\Omega\b{1}$, Algorithm \ref{alg:TSP_RNG_Edge} produces a random solution that is obtained by a $k$-exchange move from the best solution 
	found so far. 
	
	Recall that in each iteration $t,$ either $\pi_{i,j}^{t}\!=\pi_{j,i}^t\!=\pi_{\min}$ or $\pi_{i,j}^{t}\!=\!\pi_{j,i}^t\!=\!\pi_{\max}$ for any
	edge $\{i,j\}\in E.$ For convenience, we will call an edge $\{i,j\}\in E$ with $\pi_{i,j}^t\!=\!\pi_{j,i}^t\!=\!\pi_{\max}$
	a {\em high} edge, and otherwise a {\em low} edge. 
	K{\"o}tzing et al \citep{K2012Theoretical} showed the probability of the event that Algorithm \ref{alg:TSP_RNG_Edge} chooses a high edge in an arbitrary fixed step conditioned on the event
	that  $l \le \sqrt{n}$  low edges have been chosen in some $l$ steps before this step is
	$1-O(1/n).$ Our setting is only slightly different with from theirs, i.e., we use $\pi_{\min}=\frac{1}{n(n-2)}$ but they put $\pi_{\min}=\frac{1}{n(n-1)}.$ Thus, the result should
	also hold here. Claim \ref{theo:Edge_Crucial} below formally asserts this, readers may also
	refer to \citep{K2012Theoretical} for a similar proof.
	\begin{claim}\label{theo:Edge_Crucial}
		Assume $M=1,\rho=1.$ Then, the probability of choosing a high edge at any fixed step in 
		Algorithm \ref{alg:TSP_RNG_Edge} is at least $1-12/n$ if at most $\sqrt{n}$  low edges have
		been chosen before that step and there exists at least one high admissible edge  to be added.
	\end{claim}
	%
	\begin{proof}
		We now fix a step $n\!-\!m$ for some $m=0,1, \ldots, n\!-\!1,$ and assume that  $l\le \sqrt{n}$  low edges have been chosen before this
		step. Obviously, we still need to
	add \Wu{$m+1\ge 1$} edges to obtain a complete solution. We now estimate the
	numbers of admissible \Wu{high and low} edges in this step.  \WWu{Note that every of the $l$ low edges blocks at most
		$3$  of the $m\!+\!l$ remaining high edges \SecondRevision{(at most
			two which are incident to the end points of
			the low edge, and at most one that may introduce a cycle).}  So at least
		$m\!+\!l-3l=m\!-\!2l\ge m\!-\!3l$ high edges are available for adding in this step.} \Wu{ 
		Of course, it may happen that there is no admissible high edges in this step.
		However, we are not interested in such a case. We consider only 
		the case that there exists at least one admissible high edge in this step, i.e. the number of
		admissible high edges in this step is at least $\max\{1,m-3l\}.$}  Note also that the $n\!-\! m$ edges added before
	partition the  subgraph of $G=(V,E)$ with vertices $V$ and 
	edges from the partial solution constructed so far into exactly $m$ connected components (here, we see an isolated vertex also as a connected component).
	For any two of the components, there are at most $4$ admissible edges connecting them. Therefore, there
	are at most $\min \B{4\binom{m}{2}, \binom{n}{2}}$ admissible low edges. \WWu{Observing $l\le \sqrt{n}$}, the
	probability of choosing a high edge in this step is bounded from {\em below} by
	\Wu{\begin{equation}
		\label{eq:G1_edge_high1}
		1\!-\!\frac{\min\B{4\binom{m}{2},\binom{n}{2}}}{\max\{1,m-3l\}}\frac{\pi_{\min}}{\pi_{\max}}\ge
		\begin{cases}
	    1	\!-\!\frac{2m^2}{(m\!-\!3l)n(n\!-\!2)}\!\ge\!
		1\!-\!\frac{3}{(n-2)} &\text{if } m > 3\sqrt{n},\\
		1\!-\!\frac{12}{n-2} &\text{if } m\le  3\sqrt{n},
		\end{cases}
		\end{equation}}%
	where the first inequality is obtained by observing that 
	\[
	\min\B{4\binom{m}{2},\binom{n}{2}}\le\min \B{2m^2, \binom{n}{2}} \le 2m^2,
	\]
$\tau_{\max}=1-1/n, \pi_{\min}=\tfrac{1}{n(n-2)},$ and
	$
	\frac{2m^2}{m-3l} \le \frac{2}{\frac{1}{m}-\frac{3\sqrt{n}}{m^2}} \le 3n.
	$ 
	\end{proof}
	With Claim \ref{theo:Edge_Crucial}, we can show that, for any $t\in \N$ and any fixed $k\in O(1),$ the 
	probability
	of the event that  a $k$-exchange of $\X_t^{[1]}$ is produced by one application of Algorithm \ref{alg:TSP_RNG_Edge} is $\Omega(1),$ see 
	Claim \ref{theo:Lemma_RNG_Edge}.
	Here, we shall use a different proof from the one presented by K{\"otzing} et al \citep{K2012Theoretical}, 
	which appears to us as problematic. 
	\begin{claim}\label{theo:Lemma_RNG_Edge}
		Let $M=1,\rho=1.$ For any $k\in O(1),$ with probability $\Omega(1),$ 
		the random solution produced by Algorithm \ref{alg:TSP_RNG_Edge} is
		a $k$-exchange of $\X_t^{[1]}.$
	\end{claim}
	\begin{proof}
	Let $k\in O(1)$ be arbitrarily fixed, and $\mathcal{M}$ be the set of all $k$-element
	subsets of $\{1,2,\ldots,n/2\}$ (where we assume without loss of generality that
	$n$ is even). 
	Obviously, $|\mathcal{M}|\in\Theta(n^k)$ since
	$k\in O(1).$ Let $\mathbb{M}\in \mathcal{M}$ be an arbitrarily fixed $k$-element subset. 
	The probability of the
	event that Algorithm~\ref{alg:TSP_RNG_Edge} selects $k$ new edges \Wu{(low edges)} at steps \Wu{$i \in \mathbb{M}$} and $n-k$ edges \Wu{(high edges)} from $\X_t^{[1]}$ at other steps, is bounded from below by
	\begin{equation}
	\label{eq:G1_edge_high2}
	\begin{split}
	\big(1-O(\frac{1}{n})\big)^{n-k} \prod_{i\in\mathbb{M}}
	\frac{(\binom{n-i+1}{2}-(n-i+k+1))\pi_{\min}}{n(n-1)\pi_{\min}+(n-i+k+1)\pi_{\max}}
	\ge\Theta(\frac{1}{n^k}),
	\end{split}
	\end{equation}
	where $1-O(1/n)$ is a lower bound for the probability of
	selecting an edge from $\X_t^{[1]}.$ In each step $i\in \mathbb{M},$ the edges chosen before
	partition the graph into $n-i+1$ connected components, and for any two of the components there exists at least $2$ edges
	connecting them without introducing a cycle. Hence, there are at least $\binom{n-i+1}{2}$ admissible edges in each step $i\in \mathbb{M}$. Notice also that  the number of admissible
	high edges in this case is at most $n-i+k+1$ ($n-i\!+\!k\!+\!1$ is the maximal number of high edges that have not been chosen before). Therefore, each factor
	$
	\frac{(\binom{n-i+1}{2}-(n-i+k+1))\pi_{\min}}{n(n-1)\pi_{\min}+(n-i+k+1)\pi_{\max}}
	$
	of \eqref{eq:G1_edge_high2} is just
	the lower bound of the probability for choosing an admissible edge not belonging to $\X_t^{[1]}$ in a 
	step $i \in \mathbb{M}$.
	
	As a result, the probability of the random event that Algorithm \ref{alg:TSP_RNG_Edge} produces a 
	$k$-exchange of $\X_t^{[1]}$ \Wu{with $k\in O(1)$} in any of the $N$ independent draws in iteration $t\!+\!1$ is bounded from below by
	$
	|\mathcal{M}|\cdot \Theta(\frac{1}{n^k})=\Theta(n^k) \cdot\Theta(\frac{1}{n^k}) \in \Omega\b{1},
	$
	since new edges can also be added in steps $l\ge n/2.$
	
	\end{proof}
	
	Notice that in the edge-based random solution generation, 
	for any $k=2,3,\ldots,n,$ any two $k$-exchanges of $\X_t^{[1]}$ are generated with the same 
	probability, since the generation does not require adding the edges in a particular order. 
	Therefore, by Claim \ref{theo:Lemma_RNG_Edge}, for any $k\in O(1),$ any specified $k$-exchange
	of $\X_t^{[1]}$ will be produced with a probability $\Theta(1/n^k).$ Since reproducing $\X_t^{[1]}$
	can be seen as a $0$-exchange of $\X_t^{[1]},$ we can thus derive the following conclusion. 
	\begin{claim}\label{theo:Unchange}
		Let $M=1,\rho=1.$ With probability $\Omega(1),$ the random solution generated by
		Algorithm \ref{alg:TSP_RNG_Edge} has a cost not larger than that of $\X_t^{[1]}.$
	\end{claim}

	Claim \ref{theo:Edge_Gap} shows that it is unlikely that the random solution generated
	by Algorithm \ref{alg:TSP_RNG_Edge} is ``very" different from the last iteration-best
	solution $\X_t^{[1]}.$ This will be fundamental for deriving the runtime lower bound.
	\begin{claim}\label{theo:Edge_Gap}
		Let $M=1,\rho=1.$ For any $\delta\in (0,1],$ with an overwhelming
		probability $1-e^{-\omega(n^{\min\{\delta, 1/4\}/2})}$, the random solution generated by Algorithm \ref{alg:TSP_RNG_Edge}
		is a $k$-exchange move from $\X_t^{[1]}$ for some
		$k<n^{\delta}.$
	\end{claim}
	\begin{proof}
		Let $\delta\in (0,1]$ be arbitrarily fixed, and put $\gamma =\min\{\delta, 1/4\}.$ 
		To prove the claim, we just need to show that with an overwhelming probability,
		the random solution generated by Algorithm \ref{alg:TSP_RNG_Edge} is a $k$-exchange
		of $\X_t^{[1]}$ for some $k\le n^{\gamma}\le n^{1/4}.$ This is again implied by the fact
		that with an overwhelming probability, at most $n^{\gamma/2}$ low edges are chosen within the first $T:=n-\frac{3n^{\gamma}}{4}$ steps in Algorithm \ref{alg:TSP_RNG_Edge}, since
		the best case $n^{\gamma/2}+\frac{3n^{\gamma}}{4}$ is still smaller than $n^{\gamma}.$
		
		
		By Claim \ref{theo:Edge_Crucial}, for any $k\le n^{\gamma/2}$ and any
		$m\le T,$ Algorithm \ref{alg:TSP_RNG_Edge} chooses high edges with a probability 
		at least $1-12/n$ at step $m$ if at most $k$ edges have been chosen before step $m,$
		since there exist at least $n-m-3k\ge \frac{3n^\gamma}{4}-3n^{\gamma/2}\ge 3$
		 admissible high edges at step $m.$ 
		
	    Let $P$ denote the probability of the random event that at most $n^{\gamma/2}$  low edges
	    are chosen within $T$ steps, and $Q$ the probability of the random event that  at least $n^{\gamma/2}+1$ low edges
	    are chosen within the same $T$ steps. Then $P=1-Q.$ We shall bound $Q$ from above, which will give a
	    lower bound for $P.$
	    
	    Let $\EE$ be the random event that at least $n^{\gamma/2}+1$ low edges are chosen within
	    $T$ steps. Then $Q=\P[\EE].$ For each $l=1,\ldots,n^{\gamma/2}+1,$ we define a random 
	    variable $v_l$ denoting the first step $m\le T$ such that $l$ low edges are chosen within
	    $m$ steps. Obviously, $\EE$ implies the random event $\EE_1$ that $v_1<v_2<\cdots<v_{n^{\gamma/2}+1}\le T.$ Thus, $Q\le \P[\EE_1],$ and $P\ge 1-\P[\EE_1].$ 
	    
	    Observe that
	    \begin{equation*}
	    	\P[\EE_1]=\sum_{a_1<a_2<\cdots<a_{n^{\gamma/2}+1}\le T}\P[v_1=a_1,\ldots,v_{n^{\gamma/2}+1}=a_{n^{\gamma/2}+1}],
	    \end{equation*}
	    and $v_1=a_1,\ldots,v_{n^{\gamma/2}+1}=a_{n^{\gamma/2}+1}$ is equivalent
	    to the random event that before step $a_1$ only high edges are chosen, that
	    at any step between $a_l$ and $a_{l+1}$ only high edges are chosen for any $l$ with $1\le l \le n^{\gamma/2},$ and that at steps $a_1, \ldots, a_{n^{\gamma/2}+1}$ only low edges are chosen. 
	    Thus, we have by Claim \ref{theo:Edge_Crucial}
	    that
	    \[
	    \P[v_1=a_1,\ldots,v_{n^{\gamma/2}+1}=a_{n^{\gamma/2}+1}]\le \big(\frac{12}{n}\big)^{n^{\gamma/2}+1},
	    \]
	    since at each step
	    $a_l,$ there exists at least one admissible high edge and
	    we do not care about what happens after step $v_{n^{\gamma/2}+1}.$
	    
	    There are at most $\binom{T}{n^{\gamma/2}+1}$ different combinations for $a_1<a_2<\cdots<a_{n^{\gamma/2}+1}.$ Therefore,
	    $
	    P \ge 1-\P[\EE_1]\ge 1-\binom{T}{n^{\gamma/2}+1}\big(\frac{12}{n}\big)^{n^{\gamma/2}+1}.
	    $
	    
	    By Stirling's formula, and observing that $n^{\gamma/1}+1\in o(T), T\in \Theta(n),$ we have
	    $
	    \binom{T}{n^{\gamma/2}+1}\big(\frac{12}{n}\big)^{n^{\gamma/2}+1}
	    =e^{-\omega(n^{\gamma/2})}.
	   $
	    Hence, $P\ge 1-e^{-\omega(n^{\gamma/2})}$ is overwhelmingly large.
	\end{proof}
}

\section{Main results}
\label{sec:RunTime_PCE}
\rightnote{Starting here, you still did not observe many of my English correction and repeated old mistakes. I changed them again.\\
	\Wu{Thank you and sorry for that! But I used the sources you sent me last time.} } 
We shall now analyze the stochastic runtime of our two different random solution generation
methods for two classes of TSP instances that have been well studied in 
the literature.
\subsection{Stochastic runtime analysis for simple instances}

We first consider a class of simple TSP instances that is defined by the following distance function $d: E \rightarrow \R$ on a graph with $n$ vertices.\begin{equation}
\label{eq:G1_def}
d(\{i,j\})=
\begin{cases}
1&\text{if } \{i,j\}=\{i,i+1\}\text{ for each } i=1,2,\ldots,n-1,\\
1&\text{if } \{i,j\}=\{n,1\},\\
n&\text{otherwise}.
\end{cases}
\end{equation}
Obviously, TSP instances with this distance function have a unique optimal solution $s^*=(1,2,\ldots,n)$ \WWu{(in the sense of the underlying Hamiltonian cycle),} 
and $s^*$ has a cost of $n$.
The cost of an arbitrary feasible solution $s$ equals $k+(n-k)\cdot n,$ 
where $k$ is the number of edges $\in s$ that are also in $s^*.$ We shall refer to these instances as  $G_1$  in
the sequel.

The class $G_1$ has been used in \citep{Zhou2009Runtime} and \citep{K2012Theoretical}
for analyzing the expected runtime of variants of $\mathcal{MMAS}.$  Zhou \citep{Zhou2009Runtime}
proved that the $(1+1)$ MMAA algorithm has an expected runtime of $O(n^6+\frac{n\ln n}{\rho})$ on $G_1$
in the case of non-visibility (i.e., without the greedy distance information in the sampling), and has an expected runtime of $O(n^5 + \frac{n\ln n}{\rho})$ in the case of visibility (i.e., with considering the greedy distance information in the sampling).
K{\"o}tzing et al \citep{K2012Theoretical}
continued the study in \citep{Zhou2009Runtime}. They investigated the expected runtime of $(1+1)$ MMAA and its
variant MMAS$^*_{\text{Arb}}$ on $G_1$ and other TSP instances on which both $(1+1)$ MMAA 
and MMAS$^*_{\text{Arb}}$ have exponential expected runtime. MMAS$^*_{\text{Arb}}$ {\em differs } with
$(1+1)$ MMAA only in the random solution generation. MMAS$^*_{\text{Arb}}$ uses Algorithm \ref{alg:TSP_RNG_Edge} as
its random solution generation method, while $(1+1)$ MMAA used Algorithm~\ref{alg:TSP_RNG_Vertex}.
K{\"o}tzing et al \citep{K2012Theoretical} proved that MMAS$^*_{\text{Arb}}$
has an expected runtime of $O(n^3\ln n + \frac{n\ln n}{\rho})$ on $G_1.$


Theorem \ref{theo:stochastic_runtime_G1_1} shows a stochastic runtime
of $O(n^{6+\epsilon})$ for the CE variant with the add-on, i.e., Algorithm \ref{alg:ce} with max-min calibration
\eqref{eq:adjustment}, the vertex-based random solution generation,
and a stochastic runtime of $O(n^{4+\epsilon})$ for the edge-based random solution generation. These results
are comparable with the above known expected runtimes.
Although
we are not able to get strictly superior runtimes,
our results are actually stronger and more informative.
\begin{theorem}[Stochastic runtime of Algorithm \ref{alg:ce} with max-min calibration on $G_1$]
\label{theo:stochastic_runtime_G1_1}
Assume that we set $M=1$, $\rho=1$, and use Algorithm \ref{alg:ce} with the max-min calibration \eqref{eq:adjustment} for the values $\pi_{\min}=\frac{1}{n(n-2)}, \pi_{\max}=1-\frac{1}{n}$.
\rightnote{Giving the definition of min-max calibration again is not necessary. I commented this out\\
	\Wu{Thank you!}} 
%
Then
\begin{itemize}
\item[a)] if we use the vertex-based random solution generation method (Algorithm \ref{alg:TSP_RNG_Vertex}), 
and take a sample size $N\in\Omega\b{n^{5+\epsilon}}$ for any constant $\epsilon \in (0,1),$ then with a probability
at least $1-e^{-\Omega\b{N/n^{5}}}$ the optimal solution $s^*$ can be found within $n$ iterations;
\item[b)] if we use the edge-based random solution generation method (Algorithm \ref{alg:TSP_RNG_Edge}), and take
a sample size $N\in\Omega\b{n^{3+\epsilon}}$ for a constant $\epsilon\in (0,1),$ then with a probability at least
$1-e^{-\Omega\b{N/n^{3}}},$ the optimal solution can be found within \WWu{$n$ iterations.}
\end{itemize}
\end{theorem}
\begin{proof} 
	\Wurevision{We prove the Theorem by showing that the probability of the
		random event that before the optimal solution is met, the number
		of edges shared by the iteration-best and optimal solution strictly increases is overwhelmingly large. This implies that the optimal solution is found within
		$n$ iterations, since the optimal solution has 
		only $n$ edges. Furthermore, the runtimes presented in the Theorem hold. We only discuss the case of $a),$ $b)$ follows with an almost 
		identical argument.
		
		By \citep{Zhou2009Runtime} (see also proof of Theorem \ref{theo:G_1_Small_N}), if $\X_t^{[1]}$ is not optimal,
		it has at least either a $2$-opt move or a $3$-opt move.
		Note that for $G_1,$ any $k$-opt move of the iteration-best
		solution increases
		the number of its edges shared with the optimal solution. 
		By Claim \ref{theo:Lemma_RNG_Vertex}, any $2$-opt move
		is generated by Algorithm \ref{alg:TSP_RNG_Vertex} with 
		probability $\Omega(n^{-3}),$ and any $3$-opt move is generated
		with probability $\Omega(n^{-5}).$ Thus, if $\X_t^{[1]}$ is not
		optimal, $\X_{t+1}^{[1]}$ shares more edges with the optimal solution
		than $\X_t^{[1]}$ with a probability at least \SecondRevision{
		$1-(1-n^{-5})^{N}=1-e^{-\Omega(N/n^{5})}\in 1-e^{-\Omega(n^{\epsilon})}$} if $N\in \Omega(n^{5+\epsilon})$
		for any $\epsilon>0.$ Thus, this repeatedly happens within polynomially many
		number of iterations with overwhelming probability \SecondRevision{$1-e^{-\Omega(N/n^5)}.$} This completes the proof.
		}
\end{proof} 

The stochastic runtimes of Theorem \ref{theo:stochastic_runtime_G1_1} are derived for a relatively large sample size,
namely $N=\Omega(n^{5+\epsilon})$ and $N=\Omega(n^{3+\epsilon}).$ Actually, Theorem~\ref{theo:stochastic_runtime_G1_1}
may still hold for a smaller sample size. Theorem \ref{theo:G_1_Small_N} partially asserts this.
It states that the total number of iterations required
to reach the optimal solution for both generation schemes may increase considerably if a smaller sample size is used. However, the stochastic runtime does not increase. Interestingly, one can obtain a smaller stochastic runtime with a small sample size
for the edge-based random solution generation.
\begin{theorem}[Stochastic runtime of Algorithm \ref{alg:ce} on $G_1$ for a small sample size]
\label{theo:G_1_Small_N}
Assume the conditions in Theorem \ref{theo:stochastic_runtime_G1_1}, \Wu{but set} $N\in\Omega\b{n^{\epsilon}}$ for any
$\epsilon\in (0,1)$. Then:
\begin{itemize}
\item[a)] For the vertex-based random solution generation, 
Algorithm \ref{alg:ce} finds the optimal solution $s^*$ within 
$n^{6}$  iterations with a probability of $1-e^{-\Omega\b{N}}$. 
\item[b)] For the edge-based random solution generation, Algorithm \ref{alg:ce} finds the optimal solution $s^*$ within $n^{3}\ln n$
 iterations
with a probability of $1-e^{-\Omega\b{N}}$. 
\end{itemize}
\end{theorem}

\begin{proof}[Proof of Theorem \ref{theo:G_1_Small_N}] 
	\Wurevision{The proof shares a similar idea with that of Theorem \ref{theo:stochastic_runtime_G1_1}. However, we consider here the random event  that 
	the	number of edges shared by the iteration-best and optimal
		solution does not decrease and strictly increases enough times within a specified polynomial number of iterations.}
	
	\Wurevision{For $a),$ we shall consider the first $n^6$ iterations. By Claim \ref{theo:Vertex_nonincreasing}, the number of edges shared by the iteration-best and optimal solution
		does not decrease with a probability $1-\big(1-\Omega(1)\big)^N=1-e^{-\Omega(N)}$
		($N\in \Omega(n^{\epsilon})$). Therefore, the
		number does not decrease \SecondRevision{within the first $n^6$ iterations} with probability
		\SecondRevision{$\prod_{t=0}^{n^6}(1-e^{-\Omega(N)})=1-e^{-\Omega(N)}.$} By Claim \ref{theo:Lemma_RNG_Vertex}, for every
		consecutive $n^5$ iterations, if the starting iteration-best solution
		is not optimal, then with probability $1-\big((1-n^{-5})^{N}\big)^{n^5}=1-e^{-\Omega(N)},$ the number
		will strictly increase at least once within these $n^5$ iterations. Therefore, with
		overwhelming probability $1-e^{-\Omega(N)}$, the optimal solution will be
		reached within the period of the first $n^6$ iterations, since there are $n$ many
		consecutive $n^5$ iterations within that period.  
		}

$b)$ can be proved by a similar way with $a).$ \Wurevision{We shall consider
	the first $n^3\ln n$ iterations.}
\Wurevision{By Claim \ref{theo:Lemma_RNG_Edge}, with probability 
	$1-(1-\Omega(1))^N=1-e^{-\Omega(N)},$ the number of shared edges does not
	decrease in consecutive two iterations. To complete the proof, we need
	an extra fact on $2,3$-exchanges.}


K{\"o}tzing et al \citep{K2012Theoretical} showed for MMAS$_{Arb}^*$ that if
the best solution $s_t^*$ found so far has $n\!-\!k$ edges from \Wurevision{the optimal solution} $s^*$, then the probability of the event that $s_{t+1}^*$ has at least 
$n\!-\!k+1$ edges from $s^*$, is in $\Omega\b{k/n^3}.$ We
\rightnote{simple or simpler proof?\\
	\Wu{simpler, a more easier proof}} 
shall use a different but \Wu{simpler} proof to show that this also holds in our case of iteration-best reinforcement. \Wurevision{And with this fact, if $|\X_t^{[1]}\bigcap s^*|=n\!-\!k$ for some $0<k\le n,$ then with probability $1-((1-k\cdot n^{-3})^{N})^{n^3/k}=1-e^{-\Omega(N)},$ the number of edges shared by the iteration-best solution
	and $s^*$ will strictly increase at least once within the period $[t, t+n^3/k].$
	This implies that $s^*$ is sampled within the first $n^3\ln n$ iterations with
	overwhelming probability $1-e^{-\Omega(N)}$, since $n^3\ln n$ iterations can be
	partitioned into $n$ many consecutive phases $[0, n^2), [n^2, n^2+n^3/(n-1)), [n^2+n^3/(n-1), n^2+n^3/(n-1)+n^3/(n-2)), \ldots.$ We now prove that fact.
	}

We first show that when $|\X_t^{[1]}\cap s^*|=n-k$ with $k>0,$ then there exists a
$2$-opt move or a $3$-opt move for $\X_t^{[1]}$ \Wurevision{(see also \citep{Zhou2009Runtime}
	for a similar proof).}
Assume that $\X_t^{[1]}$ contains exactly $n-k$ edges
from $s^*$ for some integer $k>0.$ 
Let \Wu{$e^*=\{i, i+1\}$} be an edge in $s^*$ but
{\em not} in $\X_t^{[1]}.$ Note that each node of the graph
is exactly incident to two edges of $s^*$ and $\X_t^{[1]}$, respectively.
Therefore there exists an edge $e_0\in \X_t^{[1]}$ incident
to $i,$  an edge $e'_0\in \X_t^{[1]}$ incident to $i\!+\!1,$
and $e_0, e'_0$ are not in $s^*$. Figure~\ref{fig:Demo_NewEdgeAdding} shows an example, where 
$e_0$ is either $\{i,u\}$ or $\{i,v\}$, and 
$e'_0$ is either $\{i\!+\!1,w\}$ or $\{i\!+\!1, y\}.$
\begin{figure}[!htb]
	\centering
	\includegraphics[scale=0.15]{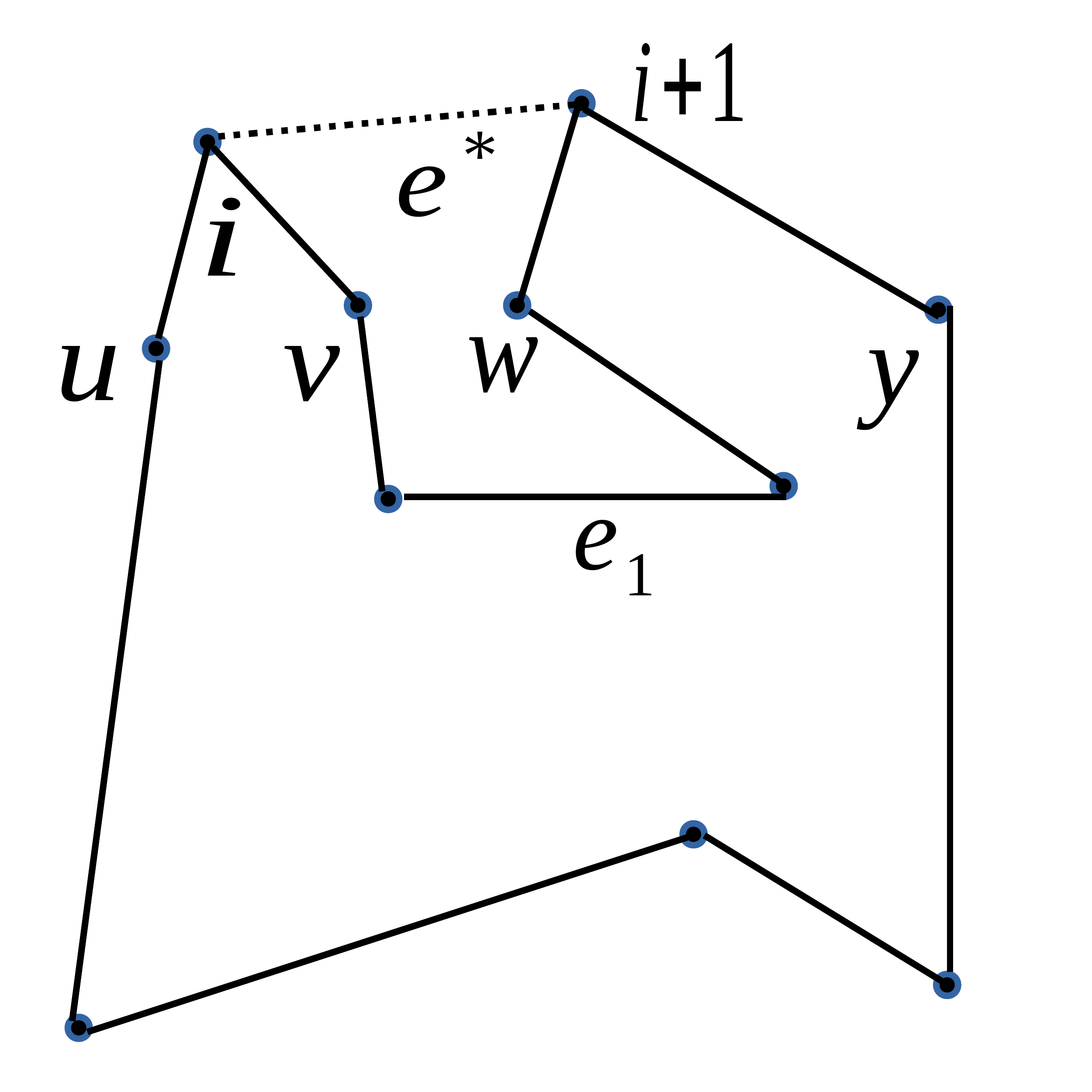}
	\caption{Demonstration of adding a new edge. The solid edges represent the cycle $\X_t^{[1]}$.}
	\label{fig:Demo_NewEdgeAdding}
\end{figure}
If $e_0=\{i,u\}$ and $e'_0=\{i\!+\!1,w\}$
or if $e_0=\{i,v\}$ and $e'_0=\{i\!+\!1,y\},$ then there exists a $2$-opt move of
\rightnote{made it clear that the improvement depends on the special distance function and added the second edge in the 2-opt move\\
	\Wu{thank you!}} 
$\X_t^{[1]}$ which removes $e_0, e'_0$ of distance $n$ and 
adds $e^*$ and another edge (either $\{u,w\}$ or $\{v,y\}$) of distance at most $n+1$ together. 
If
$e_0=\{i,u\}, e'_0=\{i\!+\!1, y\}$ or
$e_0=\{i,v\}, e'_0=\{i\!+\!1,w\},$ there is 
a $3$-opt move of $\X_t^{[1]}$ which
removes $e_0, e'_0,$ 
and an edge $e_1 \notin s^*,$ and adds edge $e^*$ and another two edges, this replacing 3 edges of distance $n$ by 3 edges of distance at most $2n\!+\!1$ together. Here,
observe the fact that adding $e^*$ to $\X_t^{[1]}$ and removing $e_0,e'_0$
from $\X_t^{[1]}$ results in graph containing a cycle, and
there must be an edge $e_1\in \X_t^{[1]}$ on that cycle that does not belong to $s^*$. We choose this edge as the edge $e_1$. Therefore, for each $e^*$ of the $k$ remaining  edges in $s^*$ that are not in $\X_t^{[1]}$, 
there exists a $2$-opt or $3$-opt move of $\X_t^{[1]}$ that
adds $e^*.$ 

\Wurevision{By Claim \ref{theo:Lemma_RNG_Edge},} 
for any $l\in O(1),$ the probability of producing an $l$-exchange of the iteration-best solution $\X_t^{[1]}$ \WWu{by Algorithm \ref{alg:TSP_RNG_Edge}}
in iteration $t+1$ is $\Omega(1).$ Since any two $l$-exchanges are produced with the
same probability, the probability of producing a {\em particular} $l$-exchange
in iteration $t+1$ is $\Omega(1/n^l).$  
As a result, Algorithm \ref{alg:TSP_RNG_Edge} produces for each edge $e^*\in s^*-\X_t^{[1]}$ a $2$-opt
or $3$-opt move of $\X_t^{[1]}$ 
that \Wu{adds} 
edge $e^*$ with probability at least
$\Omega(1/n^3)$.

Note that the generation of
a $2$-exchange (or a $3$-exchange) with two
newly added edges $e_2,e_3$ by Algorithm \ref{alg:TSP_RNG_Edge} includes two mutually 
exclusive cases ($3!$ cases
for a $3$-exchange): $e_2$ is chosen before $e_3,$ or
$e_3$ is chosen before $e_2.$ It is not difficult to
see that these two cases ($3!$ cases
for $3$-exchange) have the same probability. Therefore, 
the probability of the {\em event} that Algorithm~\ref{alg:TSP_RNG_Edge}
generates a $2$-opt or $3$-opt
move of $\X_t^{[1]}$ that $e^*$ as one of the
newly added edges and selects $e^*$ before the other newly added
edges, is bounded from below by $\Omega(1/3!n^3)=\Omega(1/n^3).$ Since $\X_t^{[1]}$ has
$k$ such $e^*$ and the corresponding $k$ events are also mutually exclusive, we obtain that the probability that
$\X_{t+1}^{[1]}$ has  more edges from $s^*$
than $\X_t^{[1]}$ if $\X_t^{[1]}$ has exactly $n-k$ edges
from $s^*$ for a constant $k>0$ is $\Omega(k/n^3)$ 

\end{proof}


\Wurevision{Corollary \ref{theo:G1_Perfect} further improves the stochastic
	runtime for an even smaller sample size. It can be proved by an argument similar to
	the proof of Theorem \ref{theo:G1_Perfect}, where we observe that \SecondRevision{$(1-(1-p(n))^{\omega(\ln n)})^{n^l}=1-n^{-\omega(1)}$ for any constant $l>0$ and probability $p(n)\in \Omega(1),$ and that
		$1-e^{-\omega(\ln n)}=1-n^{-\omega(1)}.$}
\begin{corollary}\label{theo:G1_Perfect}
	Assume the conditions in Theorem \ref{theo:stochastic_runtime_G1_1}, \Wu{but let} $N\in\omega\b{\ln n}$. Then:
	\begin{itemize}
		\item[a)] For the vertex-based random solution generation, 
		Algorithm \ref{alg:ce} finds the optimal solution $s^*$ within 
		$n^{6}$  iterations with a probability of $1-n^{-\omega(1)}$. Particularly,
		if \SecondRevision{$N=(\ln n)^{2},$} the runtime is \SecondRevision{$n^6(\ln n)^{2}$} with 
		probability $1-n^{-\omega(1)}.$
		
		\item[b)] For the edge-based random solution generation, Algorithm \ref{alg:ce} finds the optimal solution $s^*$ within $n^{3}\ln n$
		iterations
		with a probability of $1-n^{-\omega(1)}$. Particularly,
		if \SecondRevision{$N=(\ln n)^{2},$} the runtime is \SecondRevision{$n^3(\ln n)^{3}$} with 
		probability $1-n^{-\omega(1)}.$
	\end{itemize}
\end{corollary}}

\Wurevision{Theorem \ref{theo:G_1_Small_N} tells that, for any $\epsilon\in (0,1),$ a sample size of $N\in \Theta(n^{\epsilon})$ is already sufficient
	for iteration-best reinforcement to efficiently find an optimal solution of simple TSP instances with an overwhelming probability.
	Corollary \ref{theo:G1_Perfect} further shows that $N\in \omega(\ln n)$ even leads to a better runtime with a slightly smaller but still
	overwhelming probability. 
	Theorem \ref{theo:G_1_VerySmall_N} below shows that with an overwhelming
	probability, the runtime of iteration-best reinforcement will be exponential if $N\in O(\ln n)$, even if the instances are as simple as those in $G_1$.} 

\Wurevision{
\begin{theorem}
	\label{theo:G_1_VerySmall_N}
	Assume the conditions of Theorem \ref{theo:stochastic_runtime_G1_1}, \Wu{but set} $N< \frac{1}{220}\ln n$. Then, with probability  $1-e^{-\Omega\b{n^{1/200}}},$
		 Algorithm \ref{alg:ce} with edge-based solution generation
		  does not find the optimal solution $s^*$ within $e^{\Theta(n^{1/300})}$
		 iterations. 
\end{theorem}
\begin{proof}
	We prove the Theorem by inspecting the probability of the random event that,
	before the optimal solution is found, the cost of the iteration-best solution 
	$\X_t^{[1]}$ will \SecondRevision{oscillate} for exponentially many iterations with an
	overwhelming probability. We shall consider this in the last stages of
	the optimization process.
	
	Let $T_0$ be the first iteration which samples
	a solution containing at least $n-n^{1/4}+n^{1/5}$  edges from the
	optimal solution. We show that with an overwhelming probability, the number
	of common edges in the iteration-best and optimal
	solution will drop below $n-n^{1/4}+n^{1/5}$ and the
	optimal solution is not sampled before that. This will imply 
	the conclusion of Theorem \ref{theo:G_1_VerySmall_N}, since, with an overwhelming
	probability, this phenomenon
	can repeatedly occur exponentially many times before optimal solution is found.
	
	To that end, we need to show the following:
	\begin{itemize}
		\item[$1)$] For any $1/4>\delta>0,$ if $\X_t^{[1]}$ contains at least
		$n-n^{\delta}$ edges from the optimal solution, then with a probability
		$O(\frac{1}{\sqrt{n}}),$ the random solution generated by Algorithm \ref{alg:TSP_RNG_Edge}
		will contain more edges from the optimal solution than $\X_t^{[1]}$ in iteration $t+1;$
		\item[$2)$] For any $1/4>\delta>0,$ if $\X_t^{[1]}$ contains at least
		$n-n^{\delta}$  edges from the optimal solution, then with a probability
		$\Omega(1)$ (at least $e^{-5}$), the random solution generated by Algorithm
		\ref{alg:TSP_RNG_Edge} will contain fewer edges from the optimal solution than $\X_t^{[1]}$
		in iteration $t+1.$
	\end{itemize}
	However, we first use these two facts and show them afterwards.
	
	By Claim \ref{theo:Edge_Gap}, with probability $1-e^{-\omega(n^{1/20})},$
	$\X_{T_0}^{[1]}$ contains at most $n-n^{1/4}+n^{1/5}+n^{1/10}$ edges
	from the optimal solution, since the random event that the number of common edges
	from the iteration-best and optimal solution increases more than 
	$n^{1/10}$ in one iteration implies an occurrence of a $\Omega(n^{1/10})$-exchange. Similarly, by Claim \ref{theo:Edge_Gap} again,
	with probability $1-e^{-\omega(n^{1/200})},$ the iteration-best solution contains $k\in [n-n^{1/4}+n^{1/5}-n^{1/6+1/100}, n-n^{1/4}+n^{1/5}+n^{1/10}+n^{1/6+1/100}]$
	 edges from the optimal solution in each iteration $t\in [T_0, T_0+ n^{1/6}].$ 
	 This means that the optimal solution is not found in the period $[T_0, T_0+n^{1/6}]$
	 with an overwhelming probability.
	 With the help of $1)$ and $2),$ we are now to show that within this period,
	the number of edges shared by the iteration-best and optimal solution is significantly reduced with an overwhelming probability. This
	 will complete the proof.
	 
	To facilitate our discussion, we call an iteration a {\em successful iteration}
	if its iteration-best solution contains more edges from the optimal solution than the
	last iteration-best solution, and an iteration a {\em failure iteration} if its iteration-best
	solution contains fewer edges from the optimal solution than the last iteration-best solution.
	
	By $1)$ and the subsequent discussion, the expected number of successful iterations within $[T_0, T_0+n^{1/6}]$ is
	$O(\frac{\ln n}{n^{1/3}}),$ since $N< \frac{1}{220}\ln n.$ Thus, by the Chernoff bound, with probability $1-e^{-\Omega(n^{1/6})},$ 
	at most $n^{1/100}$  successful iterations can occur within $[T_0, T_0+n^{1/6}].$
	By $2)$ and the subsequent discussion, the expected number of failure iterations in $[T_0, T_0+n^{1/6}]$ is $\Omega(n^{\frac{1}{6}-\frac{1}{44}}),$
	since $N< \frac{1}{220}\ln n.$
	By the Chernoff bound, it happens that with probability $1-e^{-\Omega(n^{1/6})},$ at least $n^{1/7}$ 
	failure iterations will occur in $[T_0, T_0+n^{1/6}].$ Since  a successful iteration
	can add at most $n^{1/100}$ edges from the optimal solution with probability $1-e^{-\omega(n^{1/200})},$  it totally adds at most
	$
	n^{1/100}\times n^{1/100}=n^{1/50}
    $
	 edges from the optimal solution to the iteration-best solution within $[T_0, T_0+n^{1/6}]$ with probability $1-e^{-\omega(n^{1/200})}.$ Note that within
	$[T_0, T_0+n^{1/6}],$ with probability $1-e^{-\Omega(n^{1/6})},$ at least
	$
	n^{1/7}\times 1=n^{1/7}
	$
	 ``good" edges are removed from the iteration-best solution.
	Therefore, with overwhelming probability $1-e^{-\Omega(n^{1/200})},$ $\X_{T_0+n^{1/6}}^{[1]}$
	will contain at most
	\[
	n-n^{1/4}+n^{1/5}+n^{1/10}-n^{1/7}+n^{1/50}<n-n^{1/4}+n^{1/5}
	\]
	 edges from the optimal solution, since $\X_{T_0}^{[1]}$ contains at most 
	$n-n^{1/4}+n^{1/5}+n^{1/10}$  iterations with probability $1-e^{-\Omega(n^{1/20})}.$
	As a result, with probability $1-e^{-\Omega(n^{1/200})},$ the number of common edges in the iteration-best
	and optimal solution will again be smaller than $n-n^{1/4}+n^{1/5}$ in some iteration after $T_0,$ and  the optimal solution is not found before
	that. And this will repeatedly happen $e^{\Theta(n^{1/300})}$ times
	with probability $1-e^{-\Omega(n^{1/200})}.$
	
	To finish the proof, we now formally prove $1)$ and $2).$ We first consider $2).$
	By taking $k=2$ and considering the $\binom{n}{2}$ $2$-exchanges that happen in the first $n-3\sqrt{n}$ steps in the proof of Claim \ref{theo:Lemma_RNG_Edge}, one can
	show a tighter probability lower bound $\frac{1}{e^{5}}$ for producing $2$-exchanges of $\X_t^{[1]}$ by Algorithm \ref{alg:TSP_RNG_Edge}. Here, we observe
	that the
	probability of choosing a high edge at a step before $n-3\sqrt{n}$ is at least $1-3/(n-2),$
	see the proof of Claim \ref{theo:Edge_Crucial}.
	
	Note that if $2$-exchanges deleting $2$ edges from the optimal solution happen $N$ times in an iteration, then
	the iteration will be a failure iteration. By the above and
	the fact that any two $k$-exchanges happen with the same probability, a failure
	iteration then occurs with a probability at least
	\[
	\Bigg(\frac{1}{e^{5}}\frac{\binom{n-n^{\delta}}{2}}{\binom{n}{2}}\Bigg)^N\ge 
		\Bigg(\frac{1}{e^{5}}\frac{\binom{n-n^{\delta}}{2}}{\binom{n}{2}}\Bigg)^{\frac{1}{220}\ln n}\in \Omega(n^{-1/44}),
	\]
	where $\delta\in (0,1/4)$ and $N<\frac{1}{220}\ln n.$ This asserts $2).$
	
	$1)$ follows with a similar discussion. Since $\X_t^{[1]}$ is assumed to contain at least
	$n-n^{\delta}$  edges from the optimal solution for some $\delta \in (0,1/4),$
	and since $\Omega(n^{\delta})$-exchanges happen with an overwhelmingly small probability,
	we need to consider only $O(n^{\delta})$-exchanges when we estimate the probability of
	a successful iteration. For each $k\in \Omega(n^{\delta}),$ the proportion of failure
	$k$-exchanges is bounded from below by
	\[
	\frac{\binom{n-n^{\delta}}{k}}{\binom{n}{k}}=e^{-\frac{2kn^{\delta}}{n}}+o(1)
	\ge e^{-2n^{-1/2}}+o(1),
	\] 
	since $0<\delta<1/4,$ and
	$k$-exchanges removing $k$ edges shared by the iteration-best and optimal solution
	are not ``successful" $k$-exchanges. Since for any $k\in \Omega(n^{\delta}),$ any
	two $k$-exchanges happen with the same probability, and since the sum of the probabilities of
	successful and failure $k$-exchanges is smaller than $1,$ we conclude that
	successful $O(n^{\delta})$-exchanges happen with a probability smaller than
	$1-e^{2n^{-1/2}}\in O(\frac{1}{\sqrt{n}}).$ Therefore, a successful iteration
	happens with a probability $1-(1-O(\frac{1}{\sqrt{n}}))^N\in O(\frac{\ln n}{\sqrt{n}})$
	since $N< \frac{1}{220}\ln n.$
	\end{proof}
	Theorem \ref{theo:G_1_VerySmall_N} generalizes the finding of \citep{Neumann2010A}
	to simple TSP instances. It formally states that for $\rho=1,$  $N\in \Omega(\ln n)$ is necessary to efficiently find an
	optimal solution to TSP. By Theorem \ref{theo:G_1_VerySmall_N},
	Theorem \ref{theo:stochastic_runtime_G1_1}, Theorem \ref{theo:G_1_Small_N} and its Corollary
	\ref{theo:G1_Perfect}, we have clearly analyzed
	the impact of the size of $N$ on the resulting
	stochastic runtime for the simple TSP instances in
	the case of that $\rho=1$.
	$N\in \omega(\ln n)$ is sufficient to find the optimal
	solution in a stochastically polynomial runtime, and
	the degree of the polynomial may increase with $N$,
	but the probability guaranteeing the runtime is also increasing with $N$.
}
\subsection{Stochastic runtime analysis for grid instances}

Now, we consider more general TSP instances.
Herein, the $n$ vertices are positioned on an $m\times m$ grid for some integer $m\in \N_{+}.$ The vertices are positioned
in a way that no three of them are {\em collinear}. Figure \ref{fig:grid_set} gives an example of such an instance
where $m\!=\!5$ and $n\!=\!8$. 
\begin{figure}[!htb]
\centering
\includegraphics[scale=0.4]{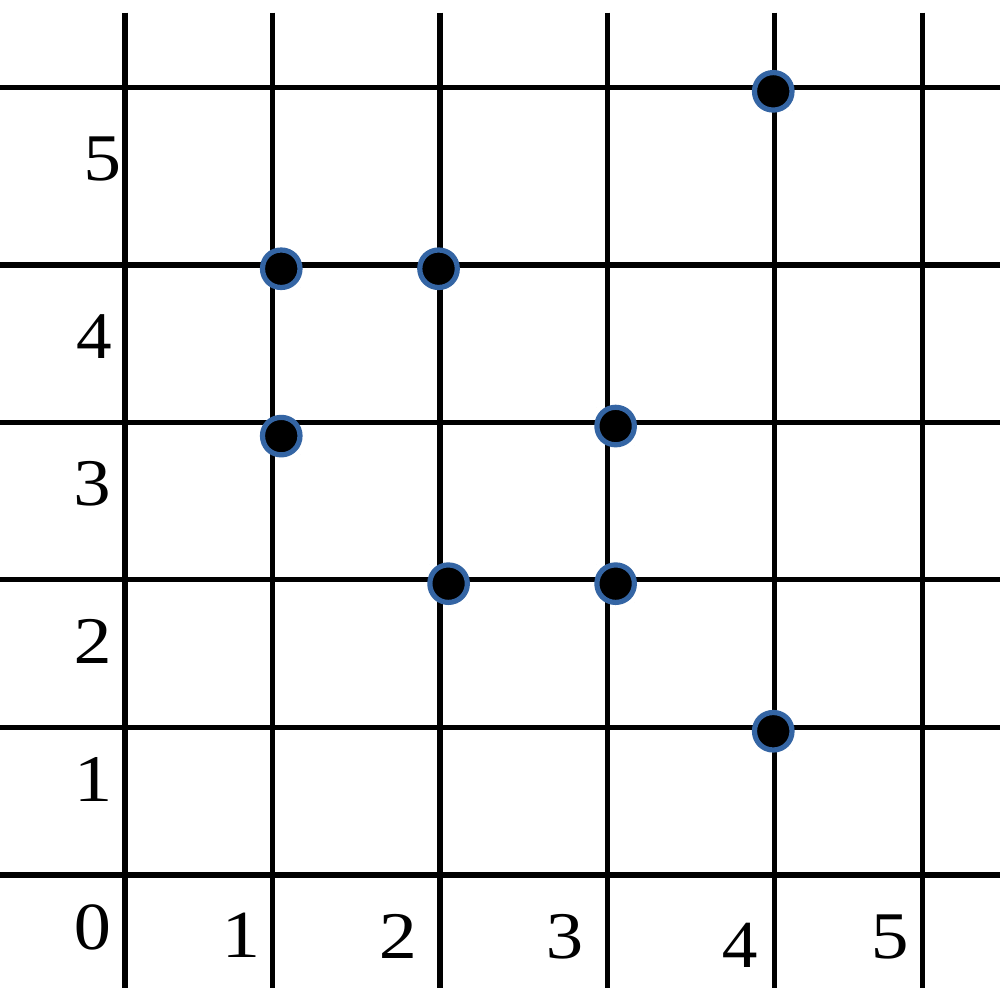}
\caption{A grid instance}
\label{fig:grid_set}
\end{figure}
The weight of an edge \Wu{$\{l,k\}\in E$} in this case is defined
\rightnote{In all figures, the letters should be of the same size or only a little smaller than in the text. In Figure 1 and 2 it is just okay, but here and in figures they are much too small\\
	\Wu{Thank you for pointing out this, I have changed the letters!}} 
as the usual Euclidean distance $d(l,k)$ between vertex $l$ and vertex $k$ for every $l,k=1,\ldots,n.$
\Wurevision{In this section,} we shall refer to these TSP instances as {\em grid instances}.



Grid instances have been studied in \citep{Sutton2012A} and \citep{Sutton2014Parameterized}.
Sutton and Neumann \citep{Sutton2012A} investigated the expected runtime of $(1\!+\!1)$ EA and RLS for these instances.
As a continuation of \citep{Sutton2012A}, Sutton et al \citep{Sutton2014Parameterized} further proved
that the more extensive algorithm $(\mu+\lambda)$ EA finds an optimal solution for the instances {\em expectedly} in 
\[
O\b{(\mu/\lambda) n^3 m^5\!+\!nm^5\!+\!(\mu/\lambda) n^{4k}(2k\!-\!1)!}
\]
  \WWu{{\em iterations}} if every of the $\lambda$ selected parents is mutated by taking a random number of consecutive 2-exchange moves, and expectedly in 
\[
O\b{(\mu/\lambda) n^3 m^5\!+\!nm^5\!+\!(\mu/\lambda) n^{2k}(k\!-\!1)!}
\] 
iterations with a mixed mutation operator, where $k$ denotes the 
number of vertices that are not on the boundary of the convex hull of $V.$ Sutton et al \citep{Sutton2014Parameterized} also studied general Euclidean TSP instances (without collinearity) and showed similar
results in terms of the maximum distance value $d_{\max},$ the minimum distance value $d_{\min},$ $k$
and the minimum angle in the triangles formed by the vertices.

Before we present our stochastic runtime, we summarize some structural properties of grid instances (some just follow from properties of general Euclidean instances). 
We say that two different edges $\{i,j\}$ and $\{k,l\}$ {\em intersect} with each other if there exists a point $p$ such that 
$p\notin \{i,j,k,l\}$ locates on both of the two edges, see, e.g., Figure \ref{fig:Intersection}. We say that a
solution is {\em intersection-free} if the corresponding Hamiltonian cycle does not contain intersections, 
see, e.g., Figure \ref{fig:IntersectionFree}.
\begin{figure}[!htb]
\centering
\begin{subfigure}{0.48\textwidth}
\centering
\includegraphics[scale=0.4]{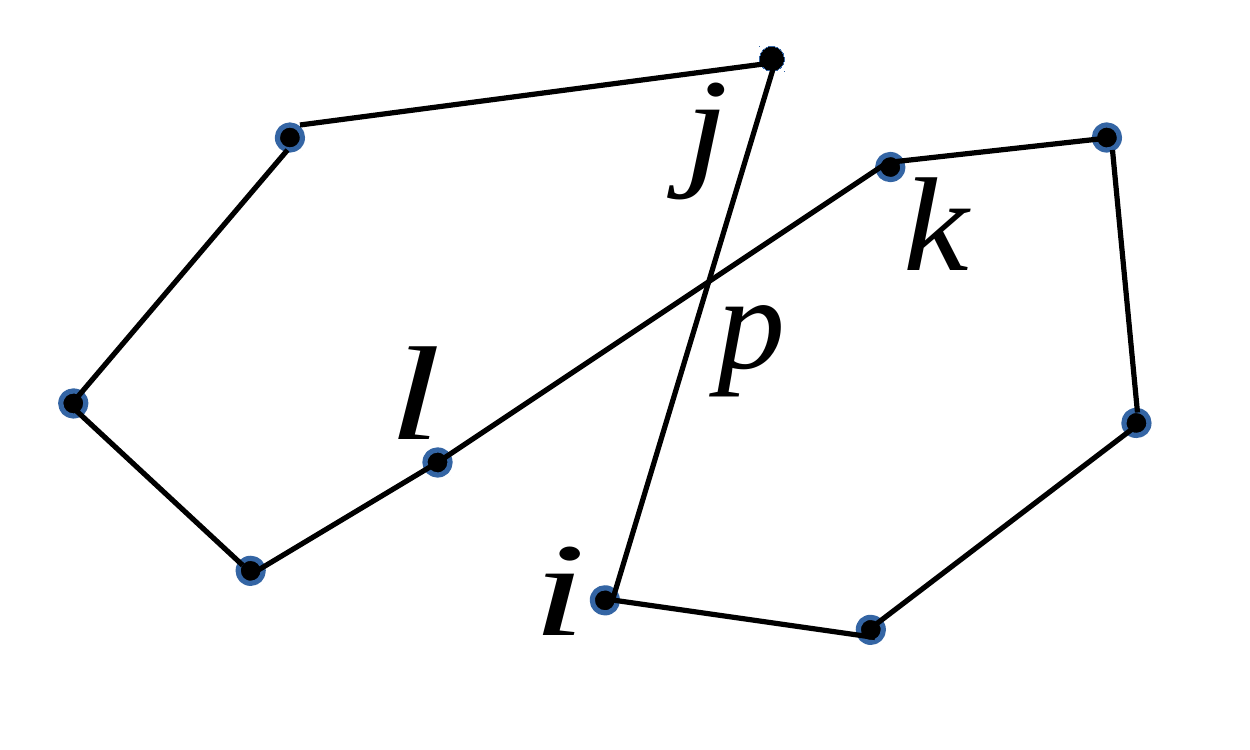}
\caption{intersection}
\label{fig:Intersection}
\end{subfigure}
\begin{subfigure}{0.48\textwidth}
\centering
\includegraphics[scale=0.4]{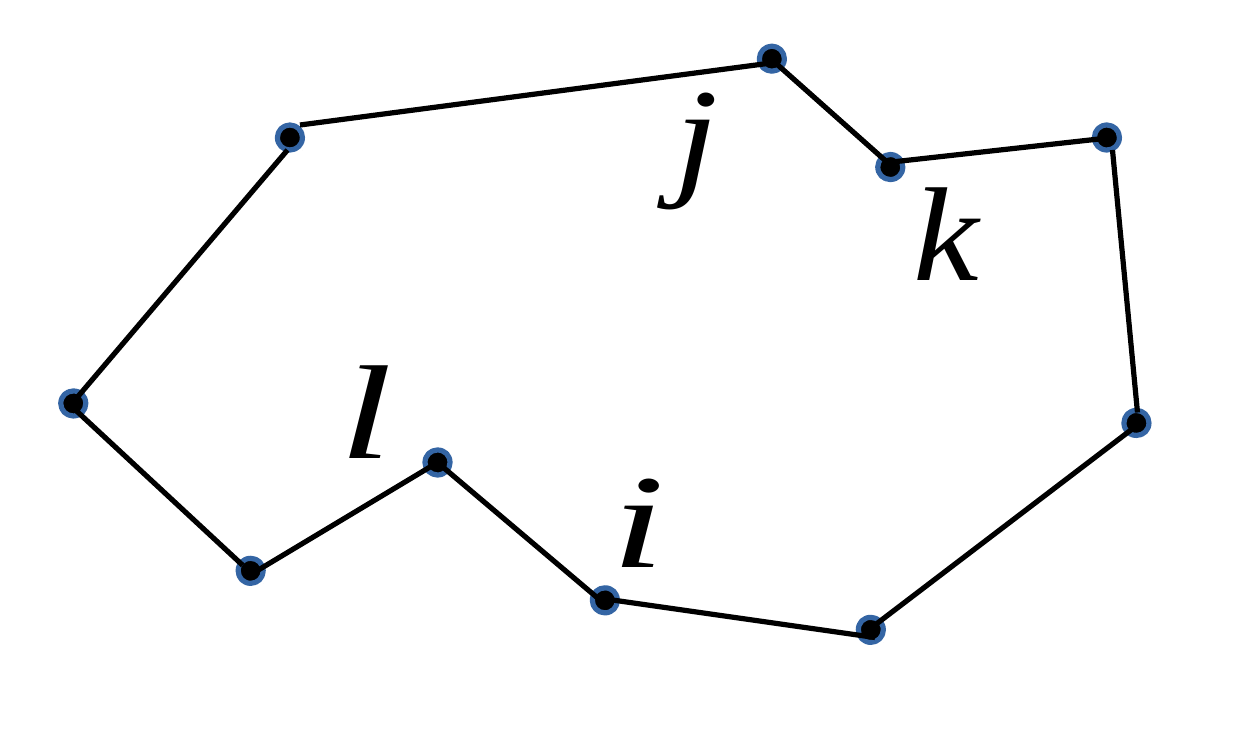}
\caption{intersection free}
\label{fig:IntersectionFree}
\end{subfigure}
\caption{Example for intersections}
\label{fig:example_intersection}
\end{figure}

Obviously, the {\em triangle inequality} \citep{Khamsi2001An} holds for grid instances. Therefore, removing
an intersection by a (unique) 2-exchange move in a solution strictly reduces the total traveling cost, see Figure \ref{fig:Intersection}. Lemma~\ref{theo:TSP_Optimality_Intersection} states the well known fact that an optimal solution of grid instances is
intersection-free.
\begin{lemma}
\label{theo:TSP_Optimality_Intersection}
Optimal solutions of grid instances are intersection-free. 
\end{lemma}

\begin{figure}[!htb]
\centering
\includegraphics[scale=0.4]{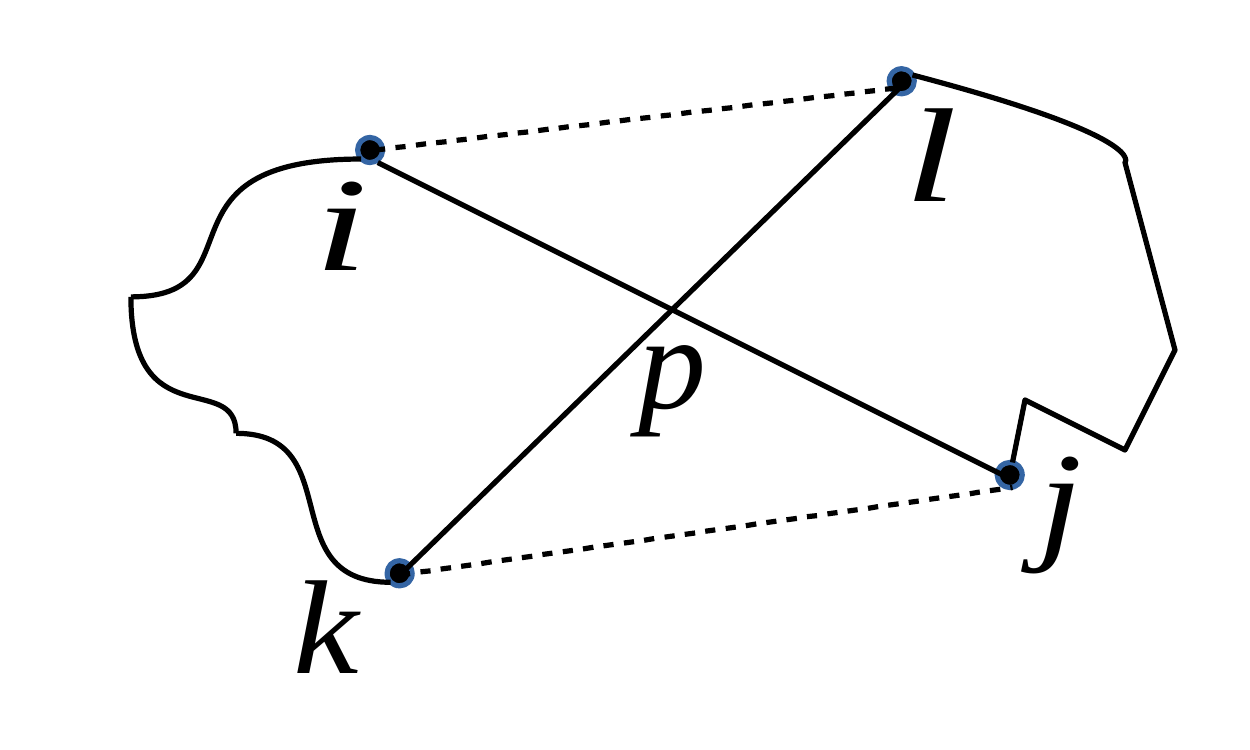}
\caption{Example for a $2$-opt move}
\label{fig:2optmove}
\end{figure}
We now {\em restrict} 2-opt moves to 2-exchange moves that remove an intersection. For
example, removing edges $\{i,j\},\{k,l\}$ in Figure \ref{fig:2optmove} and adding new edges $\{i,l\},\{k,j\}$
form such a 2-opt move.
Lemma \ref{theo:TSP_Inverse_Incremental} below says that for grid instances, removing
one intersection may reduce the total traveling cost \Wurevision{ $\Omega\b{m^{-4}}$} if it is
applicable. We omit the simple proof here. Interested readers
may refer to \citep{Sutton2014Parameterized} for a proof.
\begin{lemma}
\label{theo:TSP_Inverse_Incremental}
If a feasible solution to a grid instance contains intersections, then removing
the intersection can reduce the total traveling cost \Wurevision{$\Omega\b{m^{-4}}$.} 
\end{lemma}

\rightnote{it is convex hull, not convex hall. I had already corrected that, but you changed it back. So I changed it again...\\
	\Wu{Sorry for the mistakes!}} 
The {\em convex hull} \  $\mathfrak{Y}(V)$ of the vertex set $V$ is the smallest convex set in $\mathbb{R}^2$ that contains $V$. Its boundary is a convex polygon spanned by some vertices with possibly other vertices in the interior of that polygon. Let $V^{b}$ denote the set of vertices on the boundary of $\mathfrak{Y}(V)$.
Figure \ref{fig:convexhall} illustrates this.
\begin{figure}[!htb]
\centering
\includegraphics[scale=0.5]{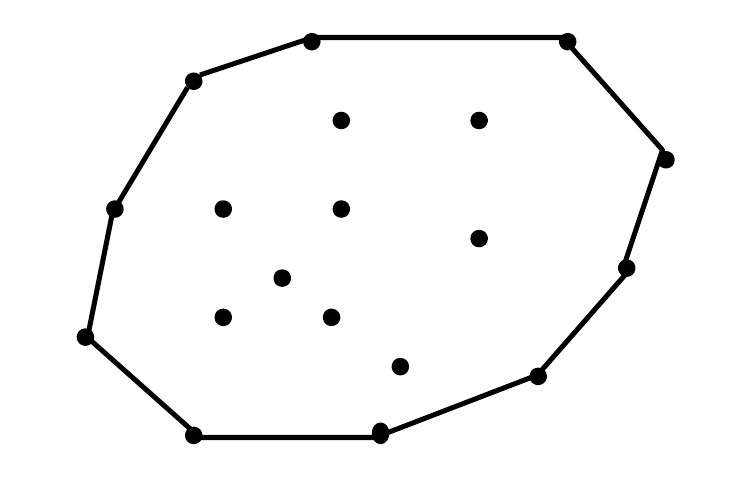}
\caption{Example of a convex hull}
\label{fig:convexhall}
\end{figure}

Quintas and Supnick \citep{Quintas1965On} proved that
if a solution $s$ is intersection-free, then the solution respects the
 hull-order, i.e., any two vertices in the subsequence of $s$ induced by the boundary (the outer polygon) of $\mathfrak{Y}(V)$ are consecutive in $s$
 if and only if they are consecutive on the boundary of $\mathfrak{Y}(V).$ Therefore, if
 $V^b=V,$ \Wurevision{i.e., all of the vertices are on the convex hull,} then every intersection-free solution is optimal.
 
  Theorem \ref{theo:TSP_Convex_Case} below analyzes the stochastic runtime of Algorithm \ref{alg:ce} for grid instances
 for the case that $V\!=\!V^b.$ It states that the
 stochastic runtime is $O(n^4\cdot m^{5+\epsilon})$ for the vertex-based random solution generation, and 
  $O(n^3 \cdot m^{5+\epsilon})$ for the edge-based random solution generation. 
  \Wurevision{Corollary \ref{theo:Corr2} further
  	improves the runtime by sacrificing the probability guarantee.} 
  These stochastic runtimes are close to the expected runtime
 $O(n^3\cdot m^5)$ for RLS reported by Sutton et al \citep{Sutton2012A} and \citep{Sutton2014Parameterized}.
 
\begin{theorem}
 \label{theo:TSP_Convex_Case}
 Consider a TSP instance with $n$ vertices located on an $m\times m$ grid such that no three of them are collinear.
 Assume that $V^b=V$, \Wurevision{i.e., every vertex in $V$ is on the convex hull $V^b$,} that we apply the max-min calibration \eqref{eq:adjustment} with $\pi_{\max}=1-\frac{1}{n}, \pi_{\min}=\frac{1}{n(n-2)}$, $\rho=1, \ M=1$ and $N\in  \Omega\b{m^\epsilon}$
 for some constant $\epsilon>0.$ Then: 
 \begin{itemize}
       \item[$a)$] With an overwhelming probability of $1-e^{-\Omega(N)},$
   Algorithm \ref{alg:ce} finds the optimal solution within 
    at most $n^4\cdot m^5$ iterations with the vertex-based random solution generation.
    
     \item[$b)$]  With an overwhelming probability
      of $1-e^{-\Omega(N)},$ Algorithm \ref{alg:ce} finds an optimal solution within 
      at most $n^3\cdot m^5$ iterations with edge-based random solution generation.
 \end{itemize}
 \end{theorem}
 
 \begin{proof}[Proof of Theorem \ref{theo:TSP_Convex_Case}]
Note that under the conditions of Theorem \ref{theo:TSP_Convex_Case}, every intersection free solution is
optimal. By Lemma \ref{theo:TSP_Inverse_Incremental}, we know that a $2$-opt move reduces the total
traveling cost by $\Omega(m^{-4}).$ Therefore, $n\cdot m^{5}$  consecutive 2-opt
moves turn a feasible solution into an optimal one, since the worst solution in this case has a 
total traveling cost smaller than $n\cdot m$ and the optimal solution has total traveling cost larger 
than $n.$ %
Notice also that $m\!\ge \!n/2,$ since the $n$ vertices are positioned on the $m\times m$ grid and
no three of them are collinear. \Wurevision{With these facts, we prove the Theorem
	by a similar argument to the one used in the proof of Theorem \ref{theo:G_1_Small_N}.}

\Wurevision{Again, we consider the random event that the cost of the iteration best solution does not
	increase within a specified period of polynomially many iterations and strictly decreases
	sufficiently many times within that period. For $a),$ we consider the first $n^4m^5$
	iterations. For $b),$ we consider the first $n^3m^5$ iterations.}

\Wurevision{For $a):$ By Claim \ref{theo:Vertex_nonincreasing}, with 
	probability $(1-(1-\Omega(1))^{N})^{n^4m^5}=1-e^{-\Omega(N)},$ the cost of the iteration-best solution does not increase within $n^4m^5$ iterations. By Claim
	\ref{theo:Lemma_RNG_Vertex}, for a phase consisting of consecutive
	$n^3$ iterations, with probability $1-(1-n^{-3})^{N\cdot n^3}=1-e^{-\Omega(N)}$, in at least one iteration  of that phase  an intersection
	is removed from the iteration-best solution, provided the phase starts with an iteration-best solution containing at least one intersection. Since the first $n^4m^5$ iterations can have
	$nm^5$ such phases, $a)$ follows.}

\Wurevision{$b)$ follows with an almost identical discussion. We therefore omit the
	proof.}
 \end{proof}
 \Wurevision{
 \begin{corollary}\label{theo:Corr2}
 	 Consider a TSP instance with $n$ vertices located on an $m\times m$ grid such that no three of them are collinear.
 	 Assume that $V^b=V$, \Wurevision{i.e., every vertex in $V$ is on the convex hull $V^b$,} that we apply the max-min calibration \eqref{eq:adjustment} with $\pi_{\max}=1-\frac{1}{n}, \pi_{\min}=\frac{1}{n(n-2)}$, $\rho=1, \ M=1$ and $N\in  \omega\b{\ln m}.$ Then: 
 	 \begin{itemize}
 	 	\item[$a)$] With probability $1-m^{-\omega(1)},$
 	 	Algorithm \ref{alg:ce} finds the optimal solution within 
 	 	at most $n^4\cdot m^5$ iterations with the vertex-based random solution generation.
 	 	
 	 	\item[$b)$]  With probability
 	 	$1-m^{-\omega(1)},$ Algorithm \ref{alg:ce} finds an optimal solution within 
 	 	at most $n^3\cdot m^5$ iterations with the edge-based random solution generation.
 	 \end{itemize}
 \end{corollary}}
 
 Now, we consider the more interesting case that \Wurevision{$|V|-|V^b|= k\in O(1),$}
 \Wurevision{i.e., $k$ vertices are not on the convex hull.} Note that 
 we can turn an arbitrary intersection-free solution to an optimal solution only by rearranging the positions of those $k$ interior points in that solution, and this requires
 at most $k$  consecutive jump moves (see \citep{Sutton2014Parameterized} for a proof). 
 A \emph{jump move} $\delta_{i,j}$ transforms a solution into another solution by shifting positions $i$, $j$ as follows. Solution
 $s$ is transformed into solution $\delta_{i,j}(s)$ by moving the vertex at position 
 $i$ into position $j$ while vertices at positions between $i$ and $j$ are shifted appropriately,
 e.g.,
 \[
 \delta_{2,5}(i_1,i_2,i_3,i_4,i_5,i_6,i_7)=(i_1,i_3,i_4,i_5,i_2,i_6,i_7)\quad\text{and}
 \]
 \[
 \delta_{5,2}(i_1,i_2,i_3,i_4,i_5,i_6,i_7)=(i_1,i_5,i_2,i_3,i_4,i_6,i_7).
 \]
 It is not difficult to see that a jump move $\delta_{i,j}$ can be simulated by either a $2$-exchange move (in the case that $|i\!-\!j|\!=\!1$)
 or a $3$-exchange move (in all other cases). Therefore, we can actually turn an intersection-free
 solution into an optimal one by a sequence of at most $k$  consecutive $2$-exchange or 
 $3$-exchange moves. Furthermore, a sequence of $k$  consecutive $2$-exchange or $3$-exchange moves can be
 simulated by a $\kappa$-exchange move with an integer $\kappa \le 3k.$ This means that any intersection-free solution can be turned into an optimal solution by
 a $\kappa$-exchange move with $\kappa\le 3k.$
  We shall call such a
 $\kappa$-exchange move in the sequel
 {\em a $3k$-opt move}, although $\kappa$ may be smaller than $3k$.
 Recall that a $3k$-opt move is produced with a probability of $\Omega\b{\frac{1}{n^{6k-1}}}$ by
 Algorithm \ref{alg:TSP_RNG_Vertex} (see \Wurevision{Claim \ref{theo:Lemma_RNG_Vertex}}), and with a probability of $\Omega\b{\frac{1}{n^{3k}}}$ by 
 Algorithm \ref{alg:TSP_RNG_Edge} (see Lemma 6 of \citep{K2012Theoretical}, \Wurevision{or
 Claim \ref{theo:Lemma_RNG_Edge}}) in \WWu{any} of
 the $N$  independent draws in iteration $t$, if $\X_{t-1}^{[1]}$ is
 intersection-free \WWu{and not optimal}. As a result, we obtain by a similar proof \WWu{as above} Theorem \ref{theo:grid_2} below.
 
 \begin{theorem}
 \label{theo:grid_2}
Consider a TSP instance with $n$ vertices located on an $m\times m$ grid such that no three of them are collinear.
 Assume that \Wurevision{$|V|-|V^b|= k\in O(1)$ ($k$ vertices are not on the convex hull $V^b$),} that we apply the max-min calibration \ref{eq:adjustment} with $\pi_{\max}=1-\frac{1}{n}, \pi_{\min}=\frac{1}{n(n-2)},$ and set $\rho=1, M=1,$
 for some constant $\epsilon>0.$ Then:
\begin{itemize} 
\item[$a)$] If we set $N\in \Omega(n^{3}\cdot m^{\epsilon}),$ then with an overwhelming probability of $1-e^{-\Omega(N/n^3)},$
   Algorithm \ref{alg:ce} finds an optimal solution within 
    at most $n\cdot m^5+n^{6k-4}$ iterations with the vertex-based random solution generation;
    
     \item[$b)$] If we set $N\in \Omega(n^{2}\cdot m^{\epsilon}),$  then with an overwhelming probability
      of $1-e^{-\Omega(N/n^2)},$ Algorithm \ref{alg:ce} finds an optimal solution within 
      at most $n\cdot m^5+n^{3k-2}$ iterations with the edge-based random solution generation.
\end{itemize}
 \end{theorem}
 
 \begin{proof}[Proof of Theorem \ref{theo:grid_2}]
We only prove $a).$ $b)$ can be derived by a very similar argument.
We define two random events as following: 
\begin{itemize}
\item[$\EE_1:$] for each $t \le n\cdot m^5+n^{6k-4},$ $f(\X_{t-1}^{[1]}) \ge f(\X_t^{[1]});$
\item[$\EE_2:$] for each $t\le n\cdot m^5+n^{6k-4},$ if $\X_{t-1}^{[1]}$ is not intersection-free, then a 2-opt move happens in iteration $t.$
\end{itemize}
By a similar argument as the one for Theorem \ref{theo:TSP_Convex_Case}, we obtain that 
$
\P\sb{\EE_1\cap \EE_2} \ge 1-e^{-\Omega\b{N/n^3}}.
$
Let $\eta$ be a random variable denoting the number of iterations for which $\X_{t-1}^{[1]}$
is intersection-free.
Notice that, conditioned on $\EE_1\cap \EE_2,$ $\eta\ge n\cdot m^5$ implies that an optimal solution
occurs within $n\cdot m^5+n^{6k-4}$ iterations.

Conditioned on $\EE_1\cap\EE_2$ and $\eta<n\cdot m^5,$ there are at least $\Omega(n^{6k-4})$ iterations in which
$\X_{t-1}$ is intersection-free, since each $\X_{t-1}^{[1]}$ is either intersection-free or not intersection-free.
Note also that in each iteration in which $\X_{t-1}^{[1]}$ intersection-free and not optimal, a $3k$-opt move that  turns $\X_{t-1}^{[1]}$ into an optimal solution  happens with probability
of at least
$
1-\b{1-\Omega\b{\frac{1}{n^{6k-1}}}}^N.$
This means for any fixed $t\in \N,$ if $\X_{t-1}^{[1]}$ is intersection-free, then the probability of the event that
$\X_t^{[1]}$ is optimal is bounded from below by
$
1-\b{1-\Omega\b{\frac{1}{n^{6k-1}}}}^N.
$
Therefore, for any fixed $\Omega\b{n^{6k-4}}$  iterations in which the iteration-best
solution $\X_{t-1}^{[1]}$ is intersection-free and \WWu{not optimal}, the probability of the
event that the corresponding $\Omega(n^{6k-4})$  $\X_t^{[1]}$'s are still not optimal,
is bounded from above by
$
\b{1-\Omega\b{\frac{1}{n^{6k-1}}}}^{N\cdot n^{6k-4}}=e^{-\Omega\b{N/n^3}}.
$
This means that, conditioned on $\EE_1\cap \EE_2$ and $\eta<n\cdot m^5,$ an optimal solution occurs within
$n\cdot m^5+n^{6k-4}$ iterations with a probability of $1-e^{-\Omega\b{N/n^3}}.$

As a result, an optimal solution occurs within the first \WWu{$n\cdot m^5+n^{6k-4}$} iterations with a probability
of $1-e^{-\Omega\b{N/n^3}}.$
 \end{proof}

 Theorem \ref{theo:grid_2} shows a stochastic runtime of $n^{3}m^{5+\epsilon}+n^{6k-1}m^{\epsilon}$
 for Algorithm \ref{alg:ce} equipped with the vertex-based solution generation, and a stochastic runtime of
 $n^{3}m^{5+\epsilon}+n^{3k}m^{\epsilon}$ for Algorithm \ref{alg:ce} equipped with edge-based solution generation,
 in the case of that $|V|-|V^b|=k\in O(1).$ This is much better than the expected runtime 
 \[
O\b{\mu\cdot n^3 m^5\!+\!nm^5\!+\!\mu\cdot n^{4k}(2k\!-\!1)!}
 \]
 for
 $(\mu\!+\!\lambda)$ EA with sequential $2$-opt mutations reported by Sutton et al  \citep{Sutton2014Parameterized}. 
 However, we are not able to analyze
 the stochastic runtime in the case that $k\in \omega\b{1},$ since $k\in \omega\b{1}$  interior points may require
 super-polynomially many iterations to turn an intersection-free solution into an optimal solution
 when a polynomial sample size is used. 
 
 


\section{Conclusion}\label{sec:conclusion}
We have analyzed the stochastic runtime of a CE algorithm on two classes of
TSP instances under two different random solution generation methods. The
stochastic runtimes are comparable with corresponding expected runtimes
reported in the literature.

Our results show that the edge-based random solution generation method
makes the algorithm more efficient for TSP instances in most cases. More-
over, \Wurevision{$N\in \Omega(\ln n)$} is \Wurevision{necessary} for
efficiently finding an optimal solution with iteration-best reinforcement. For
\SecondRevision{simple} instances, \Wurevision{$N\in \omega(\ln n)$ is sufficient to efficiently find an
optimal solution with an overwhelming probability, \Wu{and $N\in O(\ln n)$} results in an exponential runtime with an overwhelming 
probability. However,} for more difficult instances,
one may need to use a relatively large sample size.

Our stochastic runtimes are better than the
expected runtimes of the $(\mu+\lambda)$ EA on the grid instances. The EA randomly
changes local structures of some of its current solutions by a Poisson distributed number of consecutive $2$-exchange moves in every iteration, while our algorithm refrains from local operations on current solutions and only refreshes solutions by
sampling from an evolving distribution. The solution reproducing mechanism
in the EA stays the same throughout the optimization, 
only the current solutions in every iteration vary. However, the solution reproducing mechanism (sampling distribution) of our algorithm also evolves.
This is the essential difference of MBS with
 traditional EAs. The comparison of our results
with the expected runtimes in \citep{Sutton2014Parameterized} therefore show  that using a self-adaptive dynamic 
solution reproducing mechanism is helpful (in efficiently finding an optimal solution) when the search space 
\rightnote{what is rough?\\
	\Wu{same as rugged!}} 
\WWu{becomes rugged.}
The stochastic runtimes in Theorem 4 are only valid for instances with a
bounded number of interior points. In the future, it should be interesting to
analyze the case that \Wurevision{$|V|-|V^b|\in \omega\b{1}.$}
This might also give more insight
to the problem of $\mathcal{RP}\ v.s.\ \mathcal{P}$ 
\citep{Gasarch2015Classifying}.


Our analysis is actually a kind of worst-case analysis, \SecondRevision{
which is} rather pessimistic. We analyze the optimization progress by only
checking some very particular random events. This may not only underestimate
the probability \Wurevision{of finding an optimal solution with our algorithm}, but also overestimate the required number of iterations. In the future, it should be of great interest to consider
a smoothed runtime analysis over an $\epsilon$-neighborhood of the $n$ nodes
in the real plane as has been done for the Simplex method by Spielman and Teng in their
famous paper \citep{Spielman2004Smoothed}.



\Wurevision{
\section*{Acknowledgment}
We thank the anonymous reviewers for their numerous useful suggestions on 
improving the scientific quality and English presentation of this article.
}


\end{document}